\documentclass[10pt,twocloumn]{IEEEtran}

\usepackage{srcltx}

\usepackage[dvips]{graphics}
\usepackage{times}
\usepackage{indentfirst}
\usepackage{amsmath}
\usepackage{epsfig}
\usepackage{graphicx}
\usepackage{float}
\usepackage{amsfonts}
\usepackage{amssymb}
\usepackage{multirow}

\begin{document}

\newtheorem{pro}{Proposition}
\newtheorem{prop}{Property}
\newtheorem{theo}{Theorem}
\newtheorem{lem}{Lemma}
\newtheorem{rem}{Remark}

 \newcommand{\be}{\begin{equation}}
 \newcommand{\ee}{\end{equation}}
 \newcommand{\bee}{\begin{eqnarray}}
 \newcommand{\eee}{\end{eqnarray}}

 \newcommand{\nnb}{\nonumber}
 \newcommand{\qP}{ {\mathbf P}}
 \newcommand{\cf}{ {$ \,_1F_1 $}}
 \newcommand{\gf}{ {$ \,_2F_1 $}}
 \newcommand{\hf}{ {$\frac{1}{2}$}}

 \newcommand{\Ea}{ {$\frac{\Omega_1  t^2}{4}$}}
 \newcommand{\Eb}{ {$\frac{\Omega_2  t^2}{4}$}}
 \newcommand{\Ec}{ {$\frac{\Omega_3  t^2}{4}$}}
 \newcommand{\Ek}{ {$\frac{\Omega_k  t^2}{4}$}}

 \newcommand{\ppa}{ {$\sqrt{\frac{\pi \Omega_1}{4}}$}}
 \newcommand{\ppb}{ {$\sqrt{\frac{\pi \Omega_2}{4}}$}}
 \newcommand{\ppc}{ {$\sqrt{\frac{\pi \Omega_3}{4}}$}}
 \newcommand{\qqk}{ {$\sqrt{\frac{\pi \Omega_k}{4}}$}}

\newcommand{\qq}{{\mathbf q}}
\newcommand{\qT}{{\mathbf T}}

 \newcommand{\qsum}{ {$\Omega_1 + \Omega_2 +\Omega_3$}}
 \newcommand{\qii}{ {$\int_{-\infty}^{\infty}$}}
 \newcommand{\qa}{ {\mathbf a}}
 \newcommand{\qb}{ {\mathbf b}}
 \newcommand{\qc}{ {\mathbf c}}
 \newcommand{\qe}{ {\mathbf e}}

 \newcommand{\qg}{ {\mathbf g}}
 \newcommand{\qh}{ {\mathbf h}}
 \newcommand{\qm}{ {\mathbf m}}
 \newcommand{\qn}{ {\mathbf n}}

 \newcommand{\qp}{ {\mathbf p}}
 \newcommand{\qr}{ {\mathbf r}}
 \newcommand{\qs}{ {\mathbf s}}
 \newcommand{\qt}{ {\mathbf t}}
 \newcommand{\qu}{ {\mathbf u}}

 \newcommand{\qv}{ {\mathbf v}}
 \newcommand{\qw}{ {\mathbf w}}
 \newcommand{\qx}{ {\mathbf x}}
 \newcommand{\qy}{ {\mathbf y}}
 \newcommand{\qz}{ {\mathbf z}}

 \newcommand{\qA}{ {\mathbf A}}
 \newcommand{\qB}{ {\mathbf B}}
 \newcommand{\qC}{ {\mathbf C}}
 \newcommand{\qD}{ {\mathbf D}}
 \newcommand{\qE}{ {\mathbf E}}
 \newcommand{\qF}{ {\mathbf F}}
 \newcommand{\qG}{ {\mathbf G}}
 \newcommand{\qH}{ {\mathbf H}}
 \newcommand{\qI}{ {\mathbf I}}
  \newcommand{\qJ}{ {\mathbf J}}

 \newcommand{\qL}{ {\mathbf L}}
 \newcommand{\qM}{ {\mathbf M}}
 \newcommand{\qQ}{ {\mathbf Q}}

 \newcommand{\qR}{ {\mathbf R}}
 \newcommand{\qS}{ {\mathbf S}}
 \newcommand{\qX}{ {\mathbf X}}
 \newcommand{\qY}{ {\mathbf Y}}
 \newcommand{\qZ}{ {\mathbf Z}}
 \newcommand{\qU}{ {\mathbf U}}
 \newcommand{\qV}{ {\mathbf V}}
 \newcommand{\qW}{ {\mathbf W}}
 \newcommand{\qK}{ {\mathbf K}}
 \newcommand{\zR}{ {\cal R}}
 \newcommand{\qLambda}{ {\mathbf \Lambda}}
 \newcommand{\qrho}{ {$\boldsymbol \rho$}}

\newcommand{\hs}[2]{ {${\hat{#1}}{_{#2}}$}}
\newcommand{\bs}[2]{ {${\mathbf {#1}}{_{#2}}$}}
\newcommand{\hbs}[2]{ {${\hat{\mathbf {#1}}}{_{#2}}$}}
\newcommand{\hb}[1]{ {${\hat{\mathbf #1}}$}}
\renewcommand{\thesection}{\Roman{section}}

\title{Transceiver Design For SC-FDE Based MIMO Relay
Systems
\author{Peiran Wu, Robert Schober, and Vijay Bhargava}
\thanks{The authors are with the Department of Electrical and Computer
Engineering, University of British
Columbia, Vancouver, BC Canada, V6T,
1Z4, email: \{peiranw, rschober,
vijayb\}@ece.ubc.ca. This paper will be
presented in part at the IEEE Global
Communications Conference (Globecom),
Anaheim, December 2012.} } \maketitle

\begin{abstract} In this paper, we propose a joint transceiver design for single-carrier
frequency-domain equalization (SC-FDE)
based multiple-input multiple-output
(MIMO) relay systems. To this end, we
first derive the optimal minimum
mean-squared error linear and
decision-feedback frequency-domain
equalization filters at the destination
along with the corresponding error
covariance matrices at the output of
the equalizer. Subsequently, we
formulate the source and relay
precoding matrix design problem as the
minimization of a family of
Schur-convex and Schur-concave
functions of the mean-squared errors at
the output of the equalizer under
separate power constraints for the
source and the relay. By exploiting
properties of the error covariance
matrix and results from majorization
theory, we derive the optimal
structures of the source and relay
precoding matrices, which allows us to
transform the matrix optimization
problem into a scalar power
optimization problem. Adopting a high
signal-to-noise ratio approximation for
the objective function, we obtain the
global optimal solution for the power
allocation variables. Simulation
results illustrate the excellent
performance of the proposed system and
its superiority compared to
conventional orthogonal
frequency-division multiplexing based
MIMO relay systems.
\end{abstract}
\vspace*{-0.2cm}

\section{Introduction}

Multiple-input multiple-output (MIMO)
relay systems utilizing multiple
antennas at the relay node have
recently received significant research
interest due to their potential to
enhance network performance
\cite{Paulraj}. An important research
problem for MIMO relay systems is the
design of optimal node processing
matrices to improve spectral efficiency
and/or error performance through
efficient utilization of transmit
channel state information (CSIT). For
example, assuming availability of CSIT
at the source and relay nodes and
linear processing at the destination,
the source and relay processing
matrices were optimized for
maximization of the relay channel
capacity and minimization of the
mean-squared error (MSE) in \cite{Witt,
Mo} and \cite{Wu,Mo2}, respectively. In
\cite{Rong}, a general framework for
linear transceiver optimization in MIMO
relay systems was provided for a large
family of objective functions, which
includes the capacity maximizing and
the MSE minimizing designs as special
cases. The extension of the results in
\cite{Rong} to multi-hop MIMO relay
systems with linear and
decision-feedback equalization
receivers was investigated in
\cite{Rong2} and \cite{Rong3},
respectively. More recently, the design
of MIMO relay systems with partial or
imperfect CSIT at source and relay was
considered in \cite{Xing,Kim}.

Existing works on transceiver design
for MIMO relay systems are based on the
assumption of frequency-nonselective
(flat) channels
\cite{Mo,Wu,Rong2,Rong3,Kim} or
frequency-selective channels in
combination with orthogonal
frequency-division multiplexing (OFDM)
\cite{Witt,Rong,Xing}. Since OFDM
decomposes a frequency-selective
channel into multiple parallel flat
subchannels, the transceiver designs
developed for frequency-nonselective
channels can be extended to OFDM based
MIMO relay systems by solving an
additional subcarrier power allocation
problem across different subcarriers.
However, its large peak-to-average
power ratio (PAPR) makes OFDM less
appealing for application in the uplink
of wireless communication systems.
Block based single-carrier transmission
with frequency-domain equalization
(SC-FDE) is a promising alternative to
OFDM due to its comparable
implementation complexity and lower
PAPR \cite{FD-DFE,FDE}. To the best of
the authors' knowledge, the
optimization of SC-FDE based MIMO relay
systems has not been considered in the
literature so far. A key difference
between SC-FDE based MIMO systems and
MIMO-OFDM systems is that the
performance of the former depends on
the MSEs of each spatial stream whereas
the performance of the latter depends
on the subcarrier MSEs. This important
difference makes the optimization of
SC-FDE based MIMO systems more
challenging than the optimization of
MIMO-OFDM systems.

In this paper, we make the common
assumption of perfect CSI at all nodes
\cite{Witt}-\cite{Rong3} and we propose
a joint transceiver design for MIMO
relay systems employing either
frequency-domain linear equalization
(FD-LE) or frequency-domain decision
feedback equalization (FD-DFE) at the
destination. We optimize the source and
relay precoding matrices for
minimization of a general function of
the MSEs of the spatial streams under
separate power constraints for source
and relay. Specifically, as objective
functions we adopt the arithmetic MSE
(AMSE), the geometric MSE (GMSE), and
the maximum MSE (maxMSE)
\cite{Rong,Polo}, which are closely
related to channel capacity and error
rate performance. For the case of
FD-LE, we show that the optimal source
and relay precoding matrices have a
structure very similar to that of the
optimal precoding matrices in MIMO-OFDM
relay systems. However, the remaining
power allocation problem is
significantly different from the power
allocation problem for MIMO-OFDM relay
systems, especially for the GMSE and
maxMSE criteria. For FD-DFE, the
considered objective functions cannot be explicitly
expressed in terms of the optimization
variables and depend on the number of feedback filter taps, which makes a direct
solution of the optimization problem
challenging. However, we can show that for FD-DFE, the three considered objective
functions are equivalent. Furthermore,
we develop an upper bound for the objective function which is independent
of number of feedback filter taps and
is a comparatively simple function of
the optimization variables.
Interestingly, this upper bound is
shown to be identical to the GMSE
objective function for the FD-LE
receiver. Consequently, a unified
solution for the power allocation
problem for both FD-LE and FD-DFE can
be obtained, which greatly simplifies
the design procedure.

\begin{figure*}[htp]
    \centering
    \includegraphics[width=6in, height=1.4in]{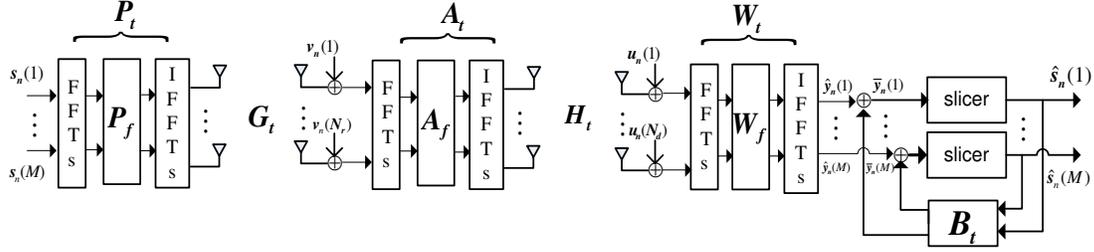}
    \caption{System model for a MIMO relay system with SC-FDE at the destination.}
{\label {fig8}}
\end{figure*}

The remainder of this paper is
organized as follows. In Section II,
the system model is presented. In
Section III, the optimal minimum MSE
(MMSE) FDE filters and the
corresponding stream error covariance
(CV) matrices are derived. The optimal
source and relay precoding matrices are
presented in Section IV. Simulation
results are given in Section V, and
some conclusions are drawn in Section
VI.

In this paper, $\rm tr(\qA)$,
$\qA^{-1}$, $\qA^{T}$, and $\qA^{\dag}$
denote the trace, inverse, transpose,
and conjugate transpose of matrix
$\qA$, respectively. $\mathbb{C}^{M
\times N}$ denotes the space of all $M
\times N$ complex matrices and $\qI_M$
is the $M\times M$ identity matrix.
$\qn \sim \mathcal{CN}(\mathbf 0 ,
\sigma_{n}^2\qI_{M})$ indicates that
$\qn\in \mathbb{C}^{M \times 1}$ is a
complex Gaussian distributed vector
with zero mean and CV matrix
$\sigma_{n}^2\qI_{M}$. $E[\cdot]$ and
$\otimes$ denote statistical
expectation and the Kronecker product,
respectively.
$\mbox{blkcirc}([\qA_1^T,\qA_2^T,...,\qA_M^T]^T)$
and
$\mbox{blkdiag}([\qA_1^T,\qA_2^T,...,\qA_M^T]^T)$
denote a block circular matrix and a
block diagonal matrix, respectively,
formed by the block-wise vector
$[\qA_1^T,\qA_2^T,...,\qA_M^T]^T$.
$\mathbb{F}_{N}$ denotes the $N \times
N$ Discrete Fourier Transform (DFT)
matrix, and $x^\star$ denotes the
optimal value of $x$.

\section{System Model\label{s2}}

We consider a block transmission system with one source node, $S$, one relay
node, $R$, and one destination node, $D$, as shown in Fig.~1. The numbers of
antennas at $S$, $R$, and $D$ are denoted by $N_s$, $N_r$, and $N_d$,
respectively. The number of spatial multiplexing data streams is
$M\leq \min\{N_s, N_r, N_d\}$. The transmission is organized in two
phases. In the first phase, $S$ processes the information symbols and
sends them to $R$. In the second phase, $R$ processes the received signal and
retransmits it to $D$. We assume there is no direct link between $S$ and $D$
due to the large pathloss and/or shadowing.

The transmit signal of each source antenna is prepended by a cyclic prefix
(CP), which comprises the last $N_{g,s}\geq L_g$ symbols of the transmitted source signal,
where $L_{g}$ denotes the largest channel impulse response (CIR) length between any
$S$-$R$ antenna pair.\footnote{For simplicity of presentation, the CP insertion is not shown in Fig.~1.}
Similarly, the transmit signal of each relay antenna is prepended by a CP, which comprises the last $N_{g,r}\geq L_{h}$
symbols of the transmitted relay signal, where $L_{h}$ is the largest CIR length between any $R$-$D$ antenna pair.

\subsection{Precoding at Source and Relay }
Let us denote the $n$th source data symbol vector as $\qs_{n}=[s_{n}(1),s_{n}(2),\ldots,s_{n}(M)]^T$,
$n=0,\ldots ,N_c-1$, where $N_c$ is the size of the data block, and $s_{n}(j)$ denotes the $n$th symbol of the $j$th
data stream, which is drawn from a constellation with variance $\sigma_{s}^2$. By stacking all $\qs_{n}$ into one vector, we obtain
$\qs=[\qs_{0}^T,\ldots, \qs_{N_c-1}^T]^T \in  \mathbb{C}^{M N_c \times 1}$. The received signal at the destination, $\qy$, can be compactly written as
\bee \label{yy}
\qy=\qH_t\qA_t\qG_t\qP_t\qs+\qH_t\qA_t\qv+\qu
\eee
with block circular matrices
\begin{align}
\qP_t&={\rm blkcirc}(
  [\qP_{t,0}^T,\cdots,\qP_{t,N_c-1}^T]^T
  ),\nnb \\
\qG_t&={\rm blkcirc}(
  [\qG_{t,0}^T,\cdots,\qG_{t,L_g-1}^T,\mathbf 0_{N_s\times N_r(N_c-L_g)}]^T ),\nnb \\
  \qA_t&={\rm blkcirc}(
  [\qA_{t,0}^T,\cdots,\qA_{t,N_c-1}^T]^T
  ),\nnb \\
\qH_t&={\rm
blkcirc}([\qH_{t,0}^T,\cdots,\qH_{t,L_h-1}^T,
\mathbf 0_{N_r\times N_d(N_c-L_h)}]^T
),\nonumber
\end{align}
where $\qP_{t,l}\in \mathbb{C}^{N_s
\times M }$, $\qG_{t,l}\in
\mathbb{C}^{N_r \times N_s}$,
$\qA_{t,l}\in \mathbb{C}^{N_r \times
N_r}$, and $\qH_{t,l}\in
\mathbb{C}^{N_d \times N_r}$ denote the
$l$th tap of the time-domain (TD)
source precoding filter, the $S$-$R$
channel, the TD relay precoding filter,
and the $R$-$D$ channel, respectively.
The noise vectors at $R$ and $D$ are
denoted by
\begin{align}
\qv&=[\qv_{0}^T,\ldots,
\qv_{N_c-1}^{T}]^{T}\sim
\mathcal{CN}(\mathbf 0 ,
\sigma_{v}^2\qI_{N_rN_c}),\nnb \\
\qu&=[\qu_{0}^T,\ldots,
\qu_{N_c-1}^{T}]^{T}\sim
\mathcal{CN}(\mathbf 0 ,
\sigma_{u}^2\qI_{N_dN_c}),
\end{align}
 where $\qv_{n}=[v_{n}(1),v_{n}(2),\ldots,v_{n}(N_r)]^T$ and $\qu_{n}=[u_{n}(1),u_{n}(2),\ldots,u_{n}(N_d)]^T$
 denote the additive white Gaussian noise (AWGN) vectors at $R$ and $D$ at time $n$, respectively.
The block circular matrices \{$\qP_t$, $\qG_t$, $\qA_t$, $\qH_t$\} can be decomposed as
\begin{align}
\qP_t&=\qF_{N_s}^\dag\qP_f\qF_{M},\quad
\qG_t=\qF_{N_r}^\dag\qG_f\qF_{N_s},\nnb
\\
\qA_t&=\qF_{N_r}^\dag\qA_f\qF_{N_r},\quad
\qH_t=\qF_{N_d}^\dag\qH_f\qF_{N_r},\quad
\end{align}
with $\qF_{\Upsilon}=
\mathbb{F}_{N_c}^\dag \otimes
\qI_{\Upsilon}$, $\Upsilon \in
\{M,N_s,N_r,N_d\}$, $\qX_f={\rm
blkdiag}([\qX_{0}^T,\cdots,\qX_{N_c-1}^T]^T)$,
and $\qX_f \in \{\qP_f, \qG_f, \qA_f,
\qH_f\}$. Here, $\qP_{k}\in
\mathbb{C}^{N_s \times M }$,
$\qG_{k}\in \mathbb{C}^{N_r \times
N_s}$, $\qA_{k}\in \mathbb{C}^{N_r
\times N_r}$, and $\qH_{k}\in
\mathbb{C}^{N_d \times N_r}$ represent
the frequency-domain (FD) source
precoding, $S$-$R$ channel, relay
precoding, and $R$-$D$ channel matrices
for the $k$th frequency tone,
respectively. We define the equivalent
end-to-end channel matrix
$\qQ_t=\qH_t\qA_t\qG_t\qP_t$ and
express it as
$\qQ_t=\qF_{N_d}^\dag\qQ_f\qF_{M}$,
where $\qQ_f={\rm
blkdiag}([\qQ_{0}^T,\cdots,\qQ_{N_c-1}^T]^T)$
with $\qQ_{k}=\qH_k\qA_k\qG_k\qP_k\in
\mathbb{C}^{N_d \times M}$ representing
the equivalent $S$-$D$ channel matrix
on the $k$th frequency tone.
Furthermore, the CV matrix of the
equivalent noise vector
$\qn=\qH_t\qA_t\qv+\qu$ can be obtained
as \bee \label{Kf}
\qK=E[\qn\qn^\dag]=\qF_{N_d}^\dag\qK_f\qF_{N_d},
\eee where
$\qK_f=\sigma_v^2\qH_f\qA_f\qA_f^\dag\qH_f^\dag+\sigma_u^2\qI_{N_d
N_c}$.

\subsection{Equalization at the Destination}

The received signal $\qy$ is
transformed into the FD using
$\qF_{N_d}$ and equalized by an FD
feedforward filter (FFF) $ \qW_{f}={\rm
blkdiag}([\qW_{0}^T,\cdots,\qW_{N_c-1}^T]^T)
$. The resulting signal is then
transformed into the TD using
$\qF_{M}^\dag $ resulting in
 \bee
\label{hat-y}\hat\qy &=& \qW_t\qy ,
\eee where
$\qW_t=\qF_{M}^\dag\qW_f\qF_{N_d}$ is
the equivalent TD FFF and $\hat
\qy=[\hat \qy_{0}^T,\ldots,\hat
\qy_{N_c-1}^{T}]^{T}$ with $\hat
\qy_{n}=$\linebreak
$[\hat y_{n}(1),\hat
y_{n}(2),\ldots,\hat y_{n}(M)]^T$
denoting the $n$th signal vector at the
output of the FFF. If FD-LE is
employed, $\hat \qy_{n}$ is the
decision variable for the $n$th source
symbol vector. On the other hand, for
FD-DFE, $\hat \qy_{n}$ is further
processed using a TD feedback filter
(FBF) to perform interference
cancelation. Assuming correct feedback
at the output of the
slicer\footnote{Correct feedback is a
common assumption for the design of
decision feedback equalizers
\cite{FD-DFE}-\cite{Dhahir}.}, the
signal corresponding to the $m$th data
stream at time $n$ at the input of the
slicer is given by \bee \label{dfe}\bar
y_{n}(m) = \hat y_n(m)-
\sum_{l=0}^{N_{fb}}[\qB_{t,l}]_{(m,:)}\qs_{(n-l){{\rm
mod} N_c}}, \eee where $\qB_{t,l}$
denotes the coefficient matrix of the
$l$th tap of the FBF, $[\qX]_{(m,:)}$
stands for the $m$th row of matrix
$\qX$, $N_{fb}$ is the number of
feedback taps, and $(\cdot){\rm mod} N$
denotes the modulo-$N$ operation. From
(\ref{dfe}) we observe that at the
initial stage of the feedback process,
i.e., when $n=0$,
$[\qs_{N_c-N_{fb}},\cdots,\qs_{N_c}]$
has to be known \emph{a priori}, which
can be accomplished by using known
training symbols. Nevertheless, for
detection of $s_{n}(m)$,
$[s_{n}(1),\cdots,s_{n}(m-1)]$ is still
unknown. Therefore, for causal
detection, the $0$th tap of the FBF,
i.e., $\qB_{t,0}$, has to be a lower
triangular matrix  with zero diagonal
entries. By collecting all $\bar
y_{n}(m)$ into a vector $\bar \qy=[\bar
\qy_{0}^T,\ldots, \bar
\qy_{N_c-1}^{T}]^{T}$ with $\bar
\qy_{n}=[\bar y_{n}(1),\bar
y_{n}(2),\ldots,\bar y_{n}(M)]^T$, we
arrive at \bee \bar\qy = \hat\qy- \qB_t
\qs, \eee where $\qB_t={\rm
blkcirc}([\qB_{t,0}^T,\cdots,\qB_{t,N_{fb}}^T,
\mathbf 0_{M\times M (N_c-N_{fb}-1)}]^T
) \in \mathbb{C}^{M N_c\times M N_c}$
is the equivalent TD-FBF. Thus, the
error vector at the input of the slicer
can be expressed as \bee \label{ee}
\qe= \bar\qy-\qs = \hat\qy-
\underbrace{(\qB_t+\qI_{M
N_c})}_{\qC_t} \qs= \hat\qy- \qC_t \qs,
\eee where $\qC_t={\rm
blkcirc}([\qC_{t,0}^T,\cdots,
  \qC_{t,N_{fb}}^T,\mathbf 0^T]^T)$ with $\qC_{t,n}=\qB_{t,n}, \forall n\neq 0$ and $\qC_{t,0}=\qB_{t,0}+\qI_M$.
  The block circular matrix $\qC_t$ can be decomposed as $\qC_t=\qF_M^\dag\qC_f\qF_M$, where $\qC_f={\rm blkdiag}
  ([\qC_{0}^T,\cdots,$ $\qC_{N_c-1}^T]^T )$. We note that by setting $\qC_t=\qI_{MN_c}$, FD-DFE reduces to FD-LE.

\section{Optimal Minimum MSE FDE Filter Design} 
In this section, we derive the optimal minimum MSE equalization filters at the
 destination and the corresponding
error CV matrices at the output of the
equalizer as functions of the source
and relay precoding matrices. Combining
(\ref{yy})-(\ref{hat-y}) and
(\ref{ee}), the error CV matrix,
$\qE\triangleq E[\qe\qe^\dag]$, can be
expressed as \bee \label{MSE0}\qE&=&
\qF_M^\dag\Big(\qW_f
(\sigma_s^2\qQ_f\qQ_f^\dag+\qK_f)\qW_f^\dag
 -\sigma_s^2\qW_f\qQ_f\qC_f^\dag\nnb \\
 &&-\sigma_s^2\qC_f\qQ_f^\dag\qW_f^\dag+\sigma_s^2\qC_f\qC_f^\dag\Big)\qF_M.
 \eee
Following the conventional equalization design methodology, the optimum FD FFF is obtained by minimizing the sum of stream MSEs,
${\rm tr}(\qE)$, which yields
\bee
\label{Wf}\qW_f^\star=\sigma_s^2 \qC_f\qQ_f^\dag\left(\sigma_s^2\qQ_f\qQ_f^\dag+\qK_f\right)^{-1}.
   \eee
Substituting $\qW_f^\star$ into (\ref{MSE0}) and simplifying the resulting expression, the CV matrix can be rewritten as
 \bee
\label{circ}\qE=\sigma_s^2\qF_M^\dag \qC_f\mathbf \Psi_f^{-1}\qC_f^\dag\qF_M,
\eee
where
$\mathbf \Psi_f= {\rm blkdiag}([\mathbf
\Psi_{0}^T,\cdots,\mathbf
\Psi_{N_c-1}^T]^T )\in \mathbb{C}^{M
N_c\times M N_c}$ with
 \begin{align}
\label{Psi}\mathbf\Psi_{k}&=\sigma_s^2\qQ_{k}^\dag\left
(\sigma_v^2\qH_k\qA_k\qA_k^\dag\qH_k^\dag+\sigma_{u}^2\qI_{N_d}\right)^{-1}\qQ_{k}
+\qI_M.
 \end{align}
From (\ref{circ}) we observe that $\qE$
is a block circular matrix. Hence, its
block diagonal entries, $\qE_n \in
\mathbb{C}^{M \times M}, \forall n$,
are identical, i.e., $\qE_n=\hat \qE,
\forall n$. Since the diagonal entries
of $\qE_n$ represent the MSEs of the
different spatial streams at time $n$,
symbols from the same stream experience
identical MSEs. CV matrix $\hat\qE$ can
be conveniently written as \bee
\label{Core-E} \hat\qE=
\frac{\sigma_s^2}{N_c}\sum_{k=1}^{N_c}
\qC_{k}^\dag\mathbf
\Psi_{k}^{-1}\qC_{k}.
 \eee
\subsection{CV Matrix and Filter Design for FD-LE}
Eqs.~(\ref{Wf}) and (\ref{Core-E}) are valid for both FD-LE and FD-DFE. For the special case of FD-LE,
we can set $\qC_f=\qI_{MN_c}$, which leads to
 \bee
\qW_f^\star=\sigma_s^2\qQ_f^\dag\left(\sigma_s^2\qQ_f\qQ_f^\dag+\qK_f\right)^{-1}
 \eee
and CV matrix
\bee
\hat\qE_{\rm FD-LE}=\frac{\sigma_s^2}{N_c}\sum_{k=1}^{N_c} \mathbf \Psi_{k}^{-1}.
 \eee
Interestingly, $\hat\qE_{\rm FD-LE}$ is equal to the arithmetic mean of the subcarrier CV matrices, $\mathbf \Psi_{k}^{-1}$, in MIMO-OFDM relay systems
\cite{Rong}.
\subsection{CV Matrix and Filter Design for FD-DFE}
The FD-DFE CV matrix depends on the
FD-FBF matrices $\qC_k$. Since the FBF
has to be implemented in the TD, we
express $\qC_k$ in terms of the TD FBF
coefficients $\qC_{t,n}$ as
$\qC_k=\sum_{n=0}^{N_{fb}}\qC_{t,n}
e^{-j\frac{2\pi}{N_c}nk}$. Now,
(\ref{Core-E}) can be rewritten as \bee
&&\hat\qE_{\rm FD-DFE}
  \nnb \\ &=&  \frac{\sigma_s^2}{N_c}\sum_{k=0}^{N_{c}-1}
 \Big[\sum_{n=0}^{N_{fb}}\qC_{t,n}  e^{-j\frac{2\pi}{N_c}nk}\mathbf \Psi_{k}^{-1}
 \sum_{m=0}^{N_{fb}}\qC_{t,m}^\dag
  e^{-j\frac{2\pi}{N_c}mk}\Big] \nnb \\
 &=&  \frac{\sigma_s^2}{N_c}\sum_{n=0}^{N_{fb}}\sum_{m=0}^{N_{fb}}\left(\qC_{t,n}\sum_{k=0}^{N_{c}-1} \mathbf \Psi_{k}^{-1}e^{-j\frac{2\pi}{N_c}(n-m)k}
 \qC_{t,m}^\dag
\right)\nnb
\\ &=&\frac{\sigma_s^2}{N_c}\hat\qC \qZ
\hat\qC^\dag. \label{FBF0} \eee To
simplify the notation, we have used the
definitions
$\hat\qC=[\qC_{t,0},...,\qC_{t,N_{fb}}]$
and \bee \qZ=\left[
      \begin{array}{cccc}
        \qz_{0} & \qz_{1} & \ldots & \qz_{N_{fb}} \\
        \qz_{1}^\dag & \qz_{0} & \ldots & \qz_{N_{fb}-1} \\
        \vdots &  & \ddots & \vdots \\
        \qz_{N_{fb}}^\dag & \qz_{N_{fb}-1}^\dag & \ldots & \qz_{0} \\
      \end{array}
    \right],\label{ZZ}
 \eee
where
$\qz_{n}=\sum_{k=0}^{N_c-1}\mathbf
\Psi_{k}^{-1}e^{j\frac{2\pi}{N_c}kn}$.
The optimal $\hat\qC$ minimizing ${\rm
tr}\{\qE_{\rm FD-DFE}\}$ can be
obtained by solving \bee \label{eq19}
\min_{\hat\qC \mathbf\Theta =\qC_{t,0}}
&& {\rm tr}(\hat\qC \qZ \hat\qC^\dag ),
\eee where
$\mathbf\Theta=[\qI_{M},\mathbf
0_{M\times (N_{fb}-1)}]$. Problem
(\ref{eq19}) can be solved using the
standard Lagrange multiplier method,
leading to \cite{Dhahir} \bee
\hat\qC^\star=\qC_{t,0}(\mathbf \Theta
\qZ^{-1}\mathbf \Theta)^{-1}\mathbf
\Theta^\dag\qZ^{-1}. \eee By
partitioning $\qZ$ and $\qZ^{-1}$ as
\bee \label{eq20}
\qZ=\left[
                       \begin{array}{cc}
                         \qZ_{11} & \qZ_{12} \\
                         \qZ_{12}^\dag & \qZ_{22} \\
                       \end{array}
                     \right], \qZ^{-1}=\left[
                       \begin{array}{cc}
                         \qU_{11} & \qU_{12} \\
                         \qU_{12}^\dag & \qU_{22} \\
                       \end{array}
                     \right],
\eee
where $\{\qZ_{11}, \qU_{11}\}\in \mathbb{C}^{M \times M}$, $\{\qZ_{12}, \qU_{12}\}\in \mathbb{C}^{M \times M N_{fb}}$, and $\{\qZ_{22} ,\qU_{22}\}\in\mathbb{C}^{M N_{fb}\times M N_{fb}}$,
$\hat\qC$ can be further expressed as
\bee
\hat\qC=[\qC_{t,0}^{-1}, \quad \qC_{t,0}^{-1}\qZ_{12}\qZ_{22}^{-1}].\label{Ct0}
\eee
Substituting (\ref{Ct0}) into (\ref{FBF0}), the FD-DFE CV matrix can be rewritten as
\bee
\qE_{\rm FD-DFE}=\frac{\sigma_s^2}{N_c} \qC_{t,0}\qU_{11}^{-1}\qC_{t,0}^\dag.
\label{E8}
\eee
To complete the FBF design, the optimal $\qC_{t,0}$ has to be determined. To this end, we introduce the Cholesky decomposition of $\qU_{11}^{-1}$ as
\bee
\label{Chol}\qU_{11}^{-1}=\qL\qD\qL^\dag,
\eee
where $\qL$ is a unit-diagonal lower triangular matrix and $\qD$ is a diagonal matrix with positive main diagonal entries. Now, it is easy to verify that the optimal $\qC_{t,0}$ which minimizes
${\rm tr}(\qE_{\rm FD-DFE})$ is given by $\qC_{t,0}^\star=\qL^{-1}$. Hence, the optimal $\hat\qC$ is obtained as
\bee
\hat\qC^\star=[\qL^{-1}, \quad \qL^{-1}\qZ_{12}\qZ_{22}^{-1}].
\eee
The structure of the optimal FBF can be interpreted  as follows: $\qL^{-1}\in \mathbb{C}^{M\times M}$ is a lower-triangular matrix which cancels the inter-antenna
interference (IAI) in the current time slot, and the remaining FBF coefficients, $\qL^{-1}\qZ_{12}\qZ_{22}^{-1}\in \mathbb{C}^{M\times MN_{fb}}$, cancel both the IAI
and inter-symbol interference (ISI) stemming from the previous $N_{fb}-1$ time slots. Inserting $\hat\qC^\star$ into (\ref{E8}) the error CV matrix can be written as
\bee
\label{mse5} \hat\qE_{\rm FD-DFE} =\frac{\sigma_s^2}{N_c} \qD.
\label{E-dfe8}
\eee
Interestingly, unlike for FD-LE, the error CV matrix for FD-DFE is a diagonal matrix, which also depends on the number of FBF taps $N_{fb}$.
\section{Source and Relay Precoding Matrix Optimization\label{s3}}
Exploiting the expressions for the error CV matrix obtained in the previous section, in this section, we minimize a general function $f(\mbox{diag}[\hat \qE])$ of the spatial stream
MSEs at the output of the equalization filter under separate constraints on the powers consumed at the source and the
relay\footnote{We note that our derivations can be extended to a joint source and relay power constraint. While such a joint power constraint offers more degrees of freedom for the system design,
separate power constraints appear more practical since usually the source node and the relay node have their own power supplies.}.
Mathematically, the optimization problem is stated as
\bee
\min_{\{\qP_k,\qA_k\}} && f(\mbox{diag}[\hat \qE]) \nnb\\
     \mbox{s.t.} && {\rm tr}\left(E[\qx\qx^\dag]\right)\leq P_S,
     \quad {\rm tr}\left(E[\qt\qt^\dag]\right)\leq P_R,  \label{prob0}
\eee
where $\hat\qE = \hat \qE_{\rm FD-LE}$ and $\hat\qE = \hat \qE_{\rm FD-DFE}$ for FD-LE and FD-DFE, respectively, $P_S$ and $P_R$ are the power budgets for $S$ and $R$, respectively, and $\mbox{diag}[\qM]$
denotes a vector containing the main diagonal entries of matrix $\qM$. The objective function $f(\mbox{diag}[\hat \qE])$ can be either a Schur-convex or
a Schur-concave increasing function with respect to (w.r.t.) $\mbox{diag}[\hat\qE]$ \cite{Polo}. For concreteness, in this paper, we consider the three most
important objective functions of this type, namely
 the arithmetic MSE (AMSE), the geometric MSE (GMSE), and the maximum MSE (maxMSE)
\bee \label{obj} f(\mbox{diag}[\hat \qE])= \begin{cases}\begin{array}{c} \sum_{m=1}^{M} \hat\qE_{mm}^{}, \quad \mbox{AMSE} \\
                                                    \prod_{m=1}^{M} \hat\qE_{mm}^{},\quad \mbox{GMSE}  \\
                                                    \max_{m=1}^{M} \hat\qE_{mm}^{}, \quad \mbox{maxMSE} \\
                                                   \end{array} \end{cases},
\eee
where $\hat \qE_{mm}^{}$ denotes the $m$th diagonal entry of $\hat\qE$. The AMSE and GMSE are Schur-concave functions while the maxMSE is a
Schur-convex function w.r.t.~$\mbox{diag}[\hat\qE]$ \cite{Polo}. We note that similar objective functions have been considered for MIMO-OFDM based
relay systems in \cite{Rong}. However, for MIMO-OFDM based relay systems, the AMSE, GMSE, and maxMSE are the sum, product, and
maximum of the subcarrier MSEs of different spatial sub-streams. In contrast, in (\ref{obj}), these three quantities are the the sum, product, and
maximum of the sub-stream MSEs of a single carrier.

The power consumptions at source and relay are given by
\begin{align}
{\rm tr}\left(E[\qx\qx^\dag]\right)&=\sigma_s^2 \sum_{k=0}^{N_c-1}{\rm tr} \left(\qP_k\qP_k^\dag\right),  \nnb \\
{\rm
tr}\left(E[\qt\qt^\dag]\right)&=\sum_{k=0}^{N_c-1}{\rm
tr}\left(\qA_k\left(\sigma_s^2\qG_k\qP_k\qP_k^\dag\qG_k^\dag
+\sigma_v^2\qI_{N_r}\right)\qA_k^\dag\right).
 \label{constr0}
\end{align}
Since the optimization variables in
(\ref{prob0}) are matrices, solving the
problem directly would incur high
complexity. In the following, we will
first derive the structure of the
optimal precoding matrices. Knowing
this structure will allow us to
transform  the optimization problem
into an optimization problem with
scalar variables.
\subsection{Structure of the Optimal Precoding Matrices for FD-LE}
We first derive the structure of the optimal source and relay precoding matrices for FD-LE. We begin by introducing the
following singular-value decompositions (SVDs) of the channel matrices
\begin{align}
\qG_k&=\qU_{G}^{(k)}\mathbf
\Lambda_{G}^{(k)} \qV_{G}^{(k)\dag}
,\quad \qH_k=\qU_{H}^{(k)}\mathbf
\Lambda_{H}^{(k)} \qV_{H}^{(k)\dag},
\quad \forall k,
 \end{align}
where $\qU_{G}^{(k)}\in
\mathbb{C}^{N_r\times N_r}$,
$\qV_{G}^{(k)} \in
\mathbb{C}^{N_s\times N_s}$ and
$\qU_{H}^{(k)}\in \mathbb{C}^{N_d\times
N_d}$, $\qV_{H}^{(k)}\in
\mathbb{C}^{N_r\times N_r}$ are the
singular-vector matrices of $\qG_k$ and
$\qH_k$, respectively. Furthermore,
$\mathbf \Lambda_{G}^{(k)} \in
\mathbb{C}^{N_r\times N_s}$ and
$\mathbf \Lambda_{H}^{(k)} \in
\mathbb{C}^{N_d\times N_r}$ are the
singular-value matrices of $\qG_k$ and
$\qH_k$, respectively, and have both
increasing main diagonal elements.

\begin{theo} For the optimization problem in (\ref{prob0}), the following structures of $\qP_k$ and $\qA_k$ are optimal
\begin{align}
\label{structure1} \qP_k^\star&=
\bar{\qV}_{G}^{(k)} \mathbf
\Lambda_{P}^{(k)}\qV_0, \quad
\qA_k^\star= \bar{\qV}_{H}^{(k)}\mathbf
\Lambda_{A}^{(k)}
\bar{\qU}_{G}^{(k)},\quad \forall k,
\end{align}
where $\bar{\qV}_{G}^{(k)}$, $\bar{\qU}_{G}^{(k)}$, and $\bar{\qV}_{H}^{(k)}$ contain the $M$ right-most columns of ${\qV}_{G}^{(k)}$,
${\qU}_{G}^{(k)}$, and ${\qV}_{H}^{(k)}$, respectively. $\mathbf \Lambda_{P}^{(k)}$ and $\mathbf \Lambda_{A}^{(k)}$ are $M\times M$ diagonal matrices with the
$m$th diagonal element denoted by $p_{km}$ and $a_{km}$, respectively.
For Schur-concave functions, $\qV_0=\qI_M$. For Schur-convex functions, $\qV_0$ is a unitary matrix chosen in such a way that all main diagonal entries of
$\hat\qE$ are equal\footnote{ In practice, $\qV_{0}$ can be chosen as a DFT matrix or a Hadamard matrix with appropriate dimensions.}.
\end{theo}
 \begin{proof} Please refer to the Appendix.
\end{proof}

Theorem 1 implies that for Schur-concave functions, the source and relay precoding matrices jointly
diagonalize the MIMO relay channels at each frequency tone, while for Schur-convex functions, the precoding
matrices diagonalize the channels up to a unitary rotation at the source. Therefore, the original optimization
problem involving matrix variables can be transformed into a scalar power optimization problem across different
spatial beams and frequency tones.
\subsection{Transformation of Optimization Problem for FD-LE}
Since the maxMSE is a Schur-convex function, according to Theorem 1, the unitary matrix, $\qV_{0}$, should be
chosen to make all diagonal entries of $\hat\qE$ equal. Recall from Section III that $\mbox{diag}[\hat\qE]$
represents the MSE of different spatial streams and all symbols of a particular stream have the same MSE. This means
that for maxMSE, identical MSE is achieved for all symbols in the SC-FDE system. Hence, the remaining maxMSE
power allocation problem is identical to that for the AMSE criterion. The only difference between the solutions for maxMSE and AMSE
minimization lies in the choice of $\qV_0$. We note that this is not true for MIMO-OFDM relay systems, where the unitary transformation
at the source only achieves identical spatial MSEs on each subcarrier, while the MSEs across subcarriers are in general different.
To balance these MSEs, multilevel waterfilling has to be carried out in such MIMO-OFDM relay systems, which entails a much higher complexity
compared to the single-level waterfilling required for the AMSE criterion, cf.~\cite{Rong}. Additionally, for MIMO-OFDM relay systems,
the unitary rotation matrices are in general different on each subcarrier as the number of transmitted data streams may
vary from subcarrier to subcarrier. However, for SC-FDE, the rotation matrices are identical for all frequency tones since the number of
data streams is determined in the time domain.

Because of the equivalence of the power allocation problems for maxMSE and AMSE, in the following, we focus on the power allocation problem for the AMSE and
GMSE criteria. From (\ref{Psi}) and (\ref{structure1}), we obtain
 $ \mathbf \Psi_{k}=\qV_0^{\dag}\mathbf {\Phi}_{k}\qV_0$ with
\bee \mathbf
{\Phi}_{k}&=&\sigma_s^2\mathbf
\Lambda_{P}^{(k)2}\bar{\mathbf
\Lambda}_{G}^{(k)2}\mathbf
\Lambda_{A}^{(k)2}\bar{\mathbf
\Lambda}_{H}^{(k)2}\nnb \\
&& \left(\sigma_v^2\mathbf
\Lambda_{A}^{(k)2}\bar{\mathbf
\Lambda}_{H}^{(k)2}+
\sigma_u^2\qI_{M}\right)^{-1}+\qI_M,
\eee where
$\bar{\mathbf\Lambda}_{G}^{(k)}$ and
$\bar{\mathbf\Lambda}_{H}^{(k)}$ are
diagonal matrices whose diagonal
entries contain the $M$ largest
singular values of ${\qG}^{(k)}$ and
${\qH}^{(k)}$, respectively. Now, we
can rewrite the objective functions as
\begin{align} f_{\rm
X}(\mathbf{\Phi})=\begin{cases}\begin{array}{c}
\sum_{m=1}^{M}
\left(\frac{1}{N_c}\sum_{k=0}^{N_c-1}{\Phi}_{km}^{-1}\right),
\mbox { X=AMSE}   \\
\sum_{m=1}^{M}\log_2
\left(\frac{1}{N_c}\sum_{k=0}^{N_c-1}{\Phi}_{km}^{-1}\right),
\mbox { X=GMSE}
\end{array}\end{cases},
 \end{align}
where
$\mathbf{\Phi}=\{{\Phi}_{km},\forall
k,m\}$ with \bee {\Phi}_{km}=
\frac{\sigma_s^2p_{km}^2 g_{km}^2
a_{km}^2 h_{km}^2} {\sigma_v^2 a_{km}^2
h_{km}^2+\sigma_u^2}+1. \eee Here,
$g_{km}$ and $h_{km}$ denote the $m$th
main diagonal elements of $\bar{\mathbf
\Lambda}_{G}^{(k)}$ and $\bar{\mathbf
\Lambda}_{H}^{(k)}$, respectively, and
represent the corresponding channel
gains of the $m$th spatial stream on
the $k$th frequency tone. Note that,
for the GMSE criterion, we have taken
the logarithm of the original objective
function to facilitate the subsequent
optimization. Due to the monotonicity
of the logarithm, the new objective
function has the same optimal solution
as the original one. The new objective
function can be rewritten as \bee
\label{CAP} -\sum_{m=1}^{M}\log_2
\left({{\rm SINR}_m}+1\right), \eee
where ${{\rm
SINR}_m}=(\frac{1}{N_c}\sum_{k=0}^{N_c-1}{\Phi}_{km}^{-1})^{-1}-1$
is the
signal-to-interference-plus-noise ratio
(SINR) of the $m$th data stream. This
implies that (\ref{CAP}) is essentially
the negative sum of the channel
capacities of different spatial
streams. Therefore, minimization of the
GMSE is equivalent to the maximization
of the capacity of the considered MIMO
SC-FDE relay system. By exploiting
(\ref{structure1}), the expression for
the power consumption on the left hand
side of the constraints in
(\ref{prob0}) can be expressed as \bee
{\rm
tr}\left(E[\qx\qx^\dag]\right)&=&\sigma_s^2
\sum_{k=0}^{N_c-1}{\rm tr}
\left(\mathbf
\Lambda_P^{(k)2}\right)=\sum_{k=0}^{N_c-1}\sum_{m=1}^{M}
P_{s,km}\nnb \\
{\rm tr}\left(E[\qt\qt^\dag]\right)&=&
 \sum_{k=0}^{N_c-1}{\rm tr}\left(\mathbf \Lambda_A^{(k)2} \left(\sigma_s^2\mathbf \Lambda_P^{(k)2}
  \bar{\mathbf  \Lambda}_G^{(k)2} +\sigma_u^2\qI_M\right)\right)
  \nnb \\&=&\sum_{k=0}^{N_c-1}\sum_{m=1}^{M}P_{r,km},
\eee where \bee \label{constr2}
P_{s,km}=\sigma_s^2 p_{km}^2, \quad
P_{r,km}= a_{km}^2(\sigma_s^2 p_{km}^2
g_{km}^2+\sigma_v^2) \eee can be
interpreted as the power allocated to
the $k$th frequency tone and the $m$th
spatial stream at the source and the
relay, respectively. By rewriting
${\Phi}_{km}$ in terms of the newly
introduced variables $P_{s,km}$ and
$P_{r,km}$ as \bee\label{eq37}
{\Phi}_{km}=\frac{P_{s,km}P_{r,km}
g_{km}^2
 h_{km}^2} {\sigma_v^2P_{r,km} h_{km}^2+\sigma_u^2
 (P_{s,km} g_{km}^2+\sigma_v^2)}+1,
 \eee
problem (\ref{prob0}) can be reformulated as the following power allocation problem
\bee
\min_{\{P_{s.km}, P_{r,km}\}} &&  f_{\rm X}\left(\mathbf\Phi\right)\nnb \\
     \mbox{s.t.}
    &&
    \sum_{k=0}^{N_c-1}\sum_{m=1}^{M}P_{s,km}\leq P_S, \nnb \\
    && \sum_{k=0}^{N_c-1}\sum_{m=1}^{M}P_{r,km}\leq
    P_R,
   \nnb \\ && P_{s,km}\geq 0,
    P_{r,km}\geq 0, \forall k,m,    \label{LE-OPA}
\eee
where the constraints $P_{s,km}\geq 0, P_{r,km}\geq 0, \forall k,m$, ensure that the allocated powers are not negative.
\subsection{Structure of the Optimal Precoding Matrices for FD-DFE}
For the FD-DFE receiver, we observe
from (\ref{E-dfe8}) that $\qE_{\rm
FD-DFE}$ is not an explicit function of
optimization variables $\qP_k$ and
$\qA_k$, which renders the optimization
a challenging task. In the following,
we will show that by using some proper
transformations, an upper bound for the
original objective function can be
derived, which is equivalent to one of
the objective functions considered for
the FD-LE receiver. However, first we will show that for FD-DFE, the three considered objective functions are equivalent.

\subsubsection{Equivalence of Objective Functions}
Since
$\qE_{\rm FD-DFE}$ in (\ref{E-dfe8}) is
a diagonal matrix, we invoke the
following matrix arithmetic-geometric
mean inequality
 \bee
 \label{key-in}\frac{1}{M} {\rm tr}(\qD)=\frac{1}{M} \sum_{i=1}^{M} [\qD]_{(i,i)}\geq
(\prod_{i=1}^{M}
[\qD]_{(i,i)})^{\frac{1}{M}}=\det(\qD)^{\frac{1}{M}}
, \eee where equality holds if and only
if (i.f.f.) all main diagonal elements
of $\qD$ are equal. The inequality
provides some important insights into
the objective function for FD-DFE.
First, it implies that the AMSE, i.e.,
${\rm tr}(\qD)$, is lower bounded by
the term involving the GMSE, i.e.,
$\det(\qD)$. Second, this lower bound
is achieved i.f.f.~the MSEs of all
streams are identical. Therefore,
making the diagonal entries of $\qD$
identical will enable us to minimize
the AMSE, GMSE, and maxMSE
simultaneously. Consequently, for
FD-DFE, the three considered objective
functions become equivalent  if this
condition is fulfilled, and the value
of the objective function is equal to
$\det(\qD)^{\frac{1}{M}}$. In the
sequel, we will show how this can be
achieved by properly choosing a
unitary matrix at the source precoder.
By noting that \bee {\rm det}(\qD)={\rm
det}(\qL\qD\qL^\dag)={\rm
det}(\qU_{11}^{-1}), \eee where we have
exploited the properties that ${\rm
det}(\qA\qB)={\rm det}(\qA){\rm
det}(\qB)$ and ${\rm det}(\qL)=1$
\cite{mat_inv}, we can obtain
\bee
\qU_{11}^{-1} &=&
\qL\qD^{1/2}(\qL\qD^{1/2})^\dag
=(\qQ\qR)^\dag\qQ\qR, \label{U11a}
\eee
where $\qQ$ is an arbitrary unitary
matrix of appropriate dimension and
$\qR=(\qL\qD^{1/2})^\dag$ is a lower
triangular matrix whose main diagonal
elements are equal to the square root
of the main diagonal elements of $\qD$.
Therefore, finding a matrix $\qD$ with
equal diagonal elements is equivalent
to finding a matrix $\qR$ with equal
diagonal elements. Since it follows
from (\ref{U11a}) that
$\qU_{11}^{-1/2}=\qQ\qR$, we have to
find a matrix decomposition of
$\qU_{11}^{-1/2}$ such that $\qR$ has
identical diagonal elements. Such a
decomposition is referred to as
equal-diagonal QR decomposition (E-QRD)
\cite{QRS} or geometric-mean
decomposition (GMD) \cite{GMD}. Note
that for an arbitrary matrix $\mathbf
\Xi$, the standard form of the GMD is,
\bee
\mathbf\Xi\qV_1^\dag=\qQ\qR,\label{GMD0}\eee
where $\qV_1$ is a unitary matrix
chosen to make the diagonal entries of
$\qR$ all equal.
 Therefore, to facilitate the application of GMD, we rewrite $\qU_{11}^{-1/2}$ as
\bee \qU_{11}^{-1/2}=\mathbf
\Xi\qV_1^\dag \label{Xi} \eee  In the
following, we will find the explicit
expression for $\mathbf \Xi$ to perform
the decomposition in (\ref{GMD0}). By
expressing $\qP_k$ as the product of
$\qV_1$ and a general matrix $\tilde
\qP_k$, \bee
\qP_k&=&\tilde\qP_k\qV_1^{\dag}, \eee
we can write $\mathbf \Psi_{k}$ as
$\mathbf \Psi_{k}=\qV_{1}\mathbf
{\hat\Psi}_{k}\qV_{1}^\dag$, where
\bee\label{eq44} \mathbf
{\hat\Psi}_{k}&=&\sigma_s^2
\tilde\qP_k^\dag\qG_k^\dag\qA_k^\dag\qH_k^\dag
\Big(\sigma_v^2\qH_k\qA_k\qA_k^\dag\qH_k^\dag+
\sigma_{u}^2\qI_{N_d}\Big)^{-1}\nnb
\\ &&\qH_k\qA_k\qG_k\tilde\qP_k+\qI_M. \eee
Note that $\mathbf {\hat\Psi}_{k}$ has
the same form of ${\mathbf \Psi}_{k}$
in (\ref{Psi}) but with $\qP_k$
replaced by $\tilde\qP_k$. Therefore,
matrix $\qZ$ in (\ref{ZZ}) can be
written as
\begin{align} \qZ&=\left[
      \begin{array}{cccc}
        \qV_1\bar\qz_{0}\qV_1^{\dag} & \qV_1\bar\qz_{1}\qV_1^{\dag} & \ldots & \qV_1\bar\qz_{N_{fb}}\qV_1^{\dag} \\
        \qV_1\bar\qz_{1}^\dag\qV_1^{\dag} & \qV_1\bar\qz_{0}\qV_1^{\dag} & \ldots & \qV_1\bar\qz_{N_{fb}-1} \\
        \vdots &  & \ddots & \vdots \\
        \qV_1\bar\qz_{N_{fb}}^\dag\qV_1^{\dag} & \qV_1\bar\qz_{N_{fb}-1}^\dag\qV_1^{\dag} & \ldots & \qV_1\bar\qz_{0}\qV_1^{\dag} \\
      \end{array}
    \right]\nnb \\&=(\qI_{N_{fb}}\otimes
\qV_{1})\bar\qZ (\qI_{N_{fb}}\otimes
\qV_{1}^\dag),\label{46}\end{align}
 where $\bar\qZ$ has the same form as $\qZ$ in (\ref{ZZ})
  with $\qz_n$ replaced by
  $\bar\qz_n={\sum_{k=0}^{N_c-1}\hat{\mathbf \Psi}_{k}^{-1}e^{j\frac{2\pi}{N_c}kn}}$.
By noting that \bee
\qZ^{-1}=(\qI_{N_{fb}}\otimes
\qV_{1})\bar\qZ^{-1}(\qI_{N_{fb}}\otimes
\qV_{1}^\dag), \eee where we have used
$(\qI_{N_{fb}}\otimes
\qV_{1})^{-1}=\qI_{N_{fb}}\otimes
\qV_{1}^{\dag}$, we obtain from (\ref{eq20})
\bee
\qU_{11}^{-1}=\qV_{1}\bar\qU_{11}^{-1}\qV_{1}^\dag,\label{U11}
\eee where $\bar\qU_{11}$ is the first
$M\times M$ submatrix of
$\bar\qZ^{-1}$. Comparing (\ref{U11})
with (\ref{Xi}), we observe that
$\bar\qU_{11}^{-1/2}$ is the explicit
form for $\mathbf \Xi$. Hence, for a
given $\bar\qU_{11}$, we can always
find a unitary matrix $\qV_1$ which
achieves the MSE lower bound.

As $\bar\qU_{11}$ is a function of
relay precoding matrix $\qA_k$ as well
as the remaining part of the source
precoding matrix, i.e., $\tilde\qP_k$,
in the following, we need to determine
these matrices. To this end, we write
$\qU_{11}=(\qZ_{11}-\qZ_{12}^\dag\qZ_{22}^{-1}\qZ_{12})^{-1}$
\cite{mat_inv}, which allows us to
express the objective function for
FD-DFE as \bee \label{50}
\mbox{OBJ}=\det(\qU_{11}^{-1})=\det\left(\qZ_{11}-\qZ_{12}^\dag\qZ_{22}^{-1}\qZ_{12}\right).
\label{dfe88} \eee

\subsubsection{Upper Bound on Objective Function}
Unfortunately, the expression for
$\mbox{OBJ}$ in (\ref{50}) depends on
the FBF length $N_{fb}$,
cf.~(\ref{46}), which is not desirable
in practice. Additionally, due to the
presence of $\qZ_{22}^{-1}$ in
(\ref{50}), it is also not
straightforward to express the
objective function in terms of $\qA_k$
and $\tilde \qP_k$. To avoid these
problems, we derive an upper bound for
$\mbox{OBJ}$, which is independent of
$N_{fb}$ and directly related to the
optimization variables.

Since matrix $\qZ$ in (\ref{ZZ}) is a
positive semidefinite (PSD) matrix,
$\qZ^{-1}$, $\qZ_{22}$, and $\qU_{11}$
are PSD matrices as well. Thus,
$\qU_{11}^{-1}$ and
$\qZ_{12}^\dag\qZ_{22}^{-1}\qZ_{12}$
are also PSD matrices. By exploiting
the fact that
$\det\left(\qA+\qB\right)\geq
\det\left(\qA\right)$ if $\qA$ and
$\qB$ are PSD matrices  \cite{mat_inv},
we obtain \bee \det(\qU_{11}^{-1})\leq
\det\left(\qZ_{11}\right)\label{inequ},
\eee where equality holds
i.f.f.~$N_{fb}=0$. Therefore, for the
case of $N_{fb}=0$,
$\det\left(\qZ_{11}\right)$ is the
exact value of $\mbox{OBJ}$. Otherwise,
it is an upper bound for $\mbox{OBJ}$,
which can be expressed as \begin{align}
\mbox{OBJ}_{ub}&={\rm det}(\qZ_{11})=
{\rm det}\left(\sum_{k=1}^{N_c} \mathbf
{\Psi}_{k}^{-1} \right) \nnb \\
&={\rm det}\left(\sum_{k=1}^{N_c}\qV_1
\mathbf {\hat\Psi}_{k}^{-1}
\qV_1^\dag\right)={\rm
det}\left(\sum_{k=1}^{N_c}\mathbf
{\hat\Psi}_{k}^{-1}
\right),\label{eq56} \end{align} where
we exploited $\det(\qV_1)=1$.

\subsubsection{Structures of Optimal Source and Relay Precoding Matrices}
Since we can always chose $\tilde
\qP_k$ such that
$\sum_{k=1}^{N_c}\mathbf
{\hat\Psi}_{k}^{-1}$ a diagonal matrix,
cf.~(\ref{eq44}), the determinant in
(\ref{eq56}) is essentially the product
of the diagonal entries of
$\sum_{k=1}^{N_c}\mathbf
{\hat\Psi}_{k}^{-1}$.  Consequently,
$\mbox{OBJ}_{ub}$ is equivalent to the
objective function for the FD-LE
receiver under the GMSE criterion.
Thus, from Theorem 1 we obtain the
following optimal structures for
$\tilde\qP_k$ and $\qA_k$ \bee
\tilde\qP_k^\star= \bar\qV_G^{(k)}
\mathbf\Lambda_{P}^{(k)},\quad
     \qA_k^\star= \bar\qV_H^{(k)}\mathbf \Lambda_{A}^{(k)} \bar\qU_{G}^{(k)\dag}.
\eee And the optimal $\qP_k^\star$ is
thus given by
$\tilde\qP_k^\star\qV_1^\dag$. The
remaining power allocation problem is
identical to that for the GMSE
criterion for FD-LE,
cf.~(\ref{LE-OPA}). It is worth
mentioning that for $N_{fb}>0$, the
upper bound $\mbox{OBJ}_{ub}$
constitutes a tight approximation of
the objective function $\mbox{OBJ}$ as
is illustrated in Section V.

\subsection{Optimal Power Allocation}
From the previous two subsections, it can be concluded that only two different types of power allocation problems have to be solved, namely the problems for the AMSE and GMSE criteria for FD-LE. The solutions
to these problems are also applicable for the maxMSE criterion for FD-LE and all three criteria for FD-DFE. However, since the objective functions for the AMSE and GMSE criteria in (\ref{LE-OPA}) are not
 jointly convex w.r.t.~the power allocation variables, the global optimal solution is difficult to obtain. Thus, in the following, we adopt a high signal-to-noise ratio (SNR) approximation for ${\Phi}_{km}$, i.e.,
 we assume $\sigma_u^2\sigma_v^2$ is sufficient small such that it can be ignored in the denominator of (\ref{eq37}), which leads to
\bee
{\Phi}_{km}\approx {\tilde\Phi}_{km}=\frac{P_{s,km}P_{r,km} g_{km}^2 h_{km}^2} {\sigma_v^2P_{r,km} h_{km}^2+\sigma_u^2 P_{s,km} g_{km}^2}+1.\label{high-phi}
 \eee
 This approximation renders the optimization problem convex such that efficient methods can be applied for its solution.  In particular, it can be shown that both objective functions $ f_{\rm AMSE}(\tilde{\mathbf\Phi})$
 and $ f_{\rm GMSE}(\tilde{\mathbf\Phi})$, $\tilde{\mathbf\Phi}=\{\tilde\Phi_{km}$ , $\forall k,m\}$, are jointly convex w.r.t.~the power allocation variables $P_{s,km}$ and $P_{r, km}$, respectively.
The convexity of the objective functions can be proved straightforwardly by showing that the Hessian matrix w.r.t.~$P_{s,km}$ and $P_{r, km}$ is positive semi-definite.
 We omit the proof here because of space constraints. Furthermore, all power constraints are affine in $P_{r,km}$ and $P_{s,km}$. Thus, the optimization problem in (\ref{LE-OPA}), with ${\Phi}_{km}$
 approximated by $\tilde{\Phi}_{km}$, is a convex optimization problem.

We are now ready to derive an iterative power allocation algorithm. To this end, we introduce the Lagrangian of the considered power allocation problem
\begin{align}
\mathcal{L} &= f_{\rm
X}({\tilde\Phi}_{km})
+\lambda[\sum_{k,m}P_{s,km}-P_S ]+\mu[
\sum_{k,m}P_{r,km}-P_R ]\nnb
\\ &-\sum_{k,m}[\beta_{km}P_{s,km}+\gamma_{km}P_{r,km}],\label{L1}
\end{align}
where $\lambda$ and $\mu$ are the
Lagrange multipliers for the sum power
constraints for source and relay,
respectively, and $\beta_{km}$ and
$\gamma_{km}$ are the Lagrange
multipliers for the individual power
constraints for source and relay,
respectively. Applying the
Karush-Kuhn-Tucker (KKT) conditions to
(\ref{L1}), which are sufficient and
necessary conditions for convex
optimization problems \cite{cvx-book},
we obtain the optimal solution of the
considered problem as \begin{align}
P^{\rm}_{s,km}&=\frac{\sigma_v^2P_{r,km}^{2}h_{km}^{2}}{g_{km}^2(P_{r,km}^{2}h_{km}^{2}+
\sigma_u^2)} \left(\sqrt{
\frac{g_{km}^2}
{\lambda (\ln2) B_m \sigma_u^2\sigma_v^2}}-1 \right)^+ \nnb \\
P_{r,km}&=\frac{\sigma_u^2P_{s,km}^{2}g_{km}^{2}}{h_{km}^2(P_{s,km}^{2}g_{km}^{2}+
\sigma_v^2)} \left(\sqrt{
\frac{h_{km}^2}{\mu (\ln2) B_m
\sigma_u^2\sigma_v^2 }}-1 \right)^+
\label{solu1}, \end{align} where
$B_m=1$ and
$B_m=\sum_{k=1}^{N_c}({\tilde\Phi}_{km}+1)^{-1}$
for the AMSE and the GMSE criteria,
respectively, and $[x]^+=\max(0,x)$.
The Lagrange multipliers $\lambda$ and
$\mu$, which are chosen to satisfy the
sum power constraint for source and
relay, respectively, can be found with
the following subgradient method
\cite{cvx-book}
\begin{align}
\lambda^{[n+1]}&=\left[\lambda^{[n]}-\varepsilon_1\left(\sum_{k=0}^{N_c-1}\sum_{m=1}^{M}P_{s,km}-P_S\right)\right]^{+} \label{lambda} \\
     \mu^{[n+1]}&=\left[\mu^{[n]}-\varepsilon_2\left(\sum_{k=0}^{N_c-1}  \sum_{m=1}^{M}P_{r,km}-P_R\right)\right]^{+}, \label{mu}
\end{align}
where $n$ is the iteration index, and
$\varepsilon_i, i=1,2$, are step sizes.
From (\ref{solu1}), we observe that the
optimal $P_{s,km}$ depends on
$P_{r,km}$ and vice versa. To tackle
this problem, we propose the algorithm
in Table I to iteratively find the
optimal power allocations. Convergence
of this algorithm to the optimal
solution is guaranteed because of the
convexity of the considered
optimization problem. Note that if
either $P_{s,km}$ or $P_{r,km}$ is
equal to $0$, the other variable will
also be $0$. This result is intuitively
pleasing since, if for example the
$(m,k)$th subchannel is shut down in
the $S$-$R$ link, there is no need to
waste power on this subchannel in the
$R$-$D$ link.
 It is also worth noting that for the GMSE criterion, $P_{s,km}$ and $P_{r,km}$ are functions of $\tilde \Phi_{km}$, which means the optimal $P_{s,km}$ and $P_{r,km}$ for the $k$th frequency tone depend
 on the power allocations in all other frequency tones. Therefore, finding the optimal solution requires a higher complexity for the GMSE criterion than for the AMSE criterion.

\begin{table}
\begin{center}
\begin{scriptsize}
\caption{Algorithm for finding the
optimal power allocation. $\epsilon_1$
and $\epsilon_2$ are small constants,
e.g.~$\epsilon_1=\epsilon_2=10^{-4}$.}
\label{exp-FDE}
\begin{tabular}{|l|l|l|}
  \hline
  1 & Initialize $\mu^{[1]}$ and $\lambda^{[1]}$\\ \hline
  2 & Initialize $P_{s,km}^{[1]}$,$P_{r,km}^{[1]}$, $\forall k,m$.\\
  & Set $P_{s,km}^{rec}=P_{s,km}^{[1]}$, $P_{r,km}^{rec}=P_{r,km}^{[1]}$,$\forall k,m$. \\\hline
  3
  &\textbf{Repeat} \\
  & \quad Set iteration number to $n=2$. \\
  &\quad\textbf{Repeat} \\
  & \quad\quad\textbf{for} $m=1:M$, $k=1:N_c$ \\
  & \quad\quad \quad Find $P_{s,km}^{[n]}$ from (\ref{solu1})
   using $P_{r,km}^{rec}$ and $\lambda^{[n-1]}$.\\
  & \quad\quad\textbf{end for}\\
  & \quad\quad Update $\lambda^{[n]}$ using (\ref{lambda}).\quad $n=n+1$.\\
  & \quad \textbf{until} $|\lambda^{[n+1]}-\lambda^{[n]}|<\epsilon_1$, \textbf{set} $P_{s,km}^{rec}=P_{s,km}^{[n]}$.\\
  & \quad Set iteration number to $l=2$. \\
  & \quad\textbf{Repeat} \\
  & \quad\quad\textbf{for} $m=1:M$, $k=1:N_c$ \\
  & \quad\quad\quad Find $P_{r,km}^{[l]}$ from (\ref{solu1}) using
   $P_{s,km}^{rec}$, and $\mu^{[l-1]}$.\\
  & \quad\quad\textbf{end for}\\
  & \quad\quad Update $\mu^{[l]}$ using (\ref{mu}).\quad $l=l+1$.\\
  & \quad\textbf{until} $|\mu^{[l+1]}-\mu^{[l]}|<\epsilon_2$, \textbf{set} $P_{r,km}^{rec}=P_{r,km}^{[n]}$.\\
  & \textbf{until} $P_{r,km}^{rec}$ and
  $P_{s,km}^{rec}$ converge.\\
 \hline
4 & $P_{s,km}^{rec}$ and
$P_{r,km}^{rec}$, $\forall k,m$, are
the optimal solution.
\\ \hline
\end{tabular}\\
\end{scriptsize}
\end{center}
\end{table}

\subsection{Suboptimal Power Allocation Schemes}
Since the proposed precoding matrix
optimization scheme involves an
iterative power allocation algorithm
and considerable information exchange
between source and relay, it is
desirable to investigate suboptimal
approaches with lower complexity and
reduced feedback overhead. One option
is to adopt equal power allocation at
the source and to optimize only the
power allocation at the relay. We refer
to the corresponding scheme as EPA-S.
EPA-S eliminates the need for
information exchange between source and
relay for power allocation, hence
guaranteeing faster convergence of the
power allocation algorithm. However,
the EPA-S scheme still requires CSI
feedback from the relay to the source
for computation of the source precoding
matrix. In order to completely avoid
CSI feedback, one can perform precoding
at the relay only, which we refer to as
ROP scheme. For FD-DFE, we also
introduce the UPS scheme, which applies
only the  unitary precoding matrix
$\qV_1$ at the source. This is
motivated by the result in Section
IV-B, where it is shown that this
unitary matrix can balance the MSEs of
the different spatial streams. Similar
to ROP, the UPS scheme has the
advantage of a reduced feedback
overhead compared to optimal power
allocation and the EPA-S scheme as the
source only needs to acquire knowledge
of the $M\times M$ unitary matrix
$\qV_0$.

\section{Simulation Results}
In this section, we evaluate the
performance of the proposed source and
relay precoding schemes using
simulations. We assume that each data
block contains $N_c=64$ symbols. The
channels are modelled as uncorrelated
Rayleigh block fading channels with
power delay profile
$p[n]=\frac{1}{\sigma_t}\sum_{l=0}^{L_x-1}
e^{-n/\sigma_t}\delta[n-l]$
\cite{Ch-book}, where $L_x\in
\{L_g,L_h\}$ and $\sigma_t=2$, which
corresponds to moderately
frequency-selective fading. For
convenience, we set the values of
$L_g$, $L_h$, $N_{g,s}$, and $N_{g,r}$
all equal to 16. Unless stated
otherwise, $N_{fb}$ is set to 15. We
assume identical noise variances for
both links, i.e.,
$\sigma_{u}^2=\sigma_{v}^2$, and define
the source and relay SNRs as
$(E_b/N_0)_s\triangleq\frac{P_S}{N_bN_s
N_c\sigma_{u}^2}$ and
$(E_b/N_0)_r\triangleq\frac{P_R}{N_b
N_s N_c\sigma_{v}^2}$, respectively,
where $N_b$ is the number of bits per
symbol. For all simulation results
shown, we set $(E_b/N_0)_s=16$ dB and
examine bit error rate (BER) and
capacity as functions of $(E_b/N_0)_r$.
All simulations are averaged over at
least 10,000 independent channel
realizations. In the following, the
proposed joint source and relay
precoding design is referred to as JSR,
and the notation $\{M,N_s,N_r,N_d\}$ is
used to specify a system with the
parameters appearing in the bracket.
\subsection{Convergence of the Algorithm and Tightness of $OBJ_{ub}$ for FD-DFE }
In Figs.~2 and 3, we show the value of
the objective function for the GMSE
criterion versus the numbers of inner
and outer iterations for different
values of $(E_b/N_0)_r$ for a
$\{2,2,2,2\}$ MIMO relay system
\footnote{Similar results also hold for
the AMSE criterion.}. We define an
outer iteration as one optimization of
$\{P_{s,km}\}$ or $\{P_{r,km}\}$ in the
algorithm shown in Table I, and the update
of $\{P_{s,km}\}$ ($\{P_{r,km}\}$) in
each outer iteration constitutes one
inner iteration.  The reference lines
indicate the optimal values of the
objective function. Figs.~2 and 3 show
that it takes at most ten inner
iterations to obtain an intermediate
solution for $\{P_{s,km}\}$ or
$\{P_{r,km}\}$, and at most three outer
iterations to obtain the final
solutions for $\{P_{s,km}\}$ and
$\{P_{r,km}\}$. Also, from Fig.~3, we
observe that the largest improvement of
the objective function value is
obtained in the first and second outer
iterations when $(E_b/N_0)_r$ is small
and large, respectively. This
observation suggests that for low SNR,
optimizing the source or the relay
power is sufficient to realize most of
the achievable performance gain, while
for high SNR, a joint optimization of
the source and relay powers is
beneficial.

\begin{figure}[tp]
    \centering
    \includegraphics[width=3.5in]{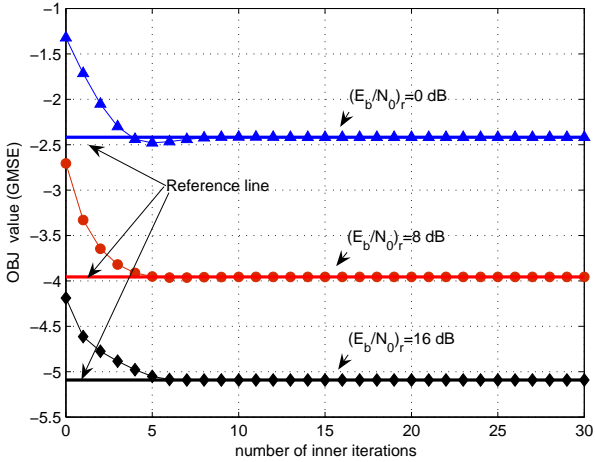}
    \caption{Objective function value (GMSE criterion) versus number of inner iterations.}
{\label {}}
\end{figure}

\begin{figure}[tp]
    \centering
    \includegraphics[width=3.5in]{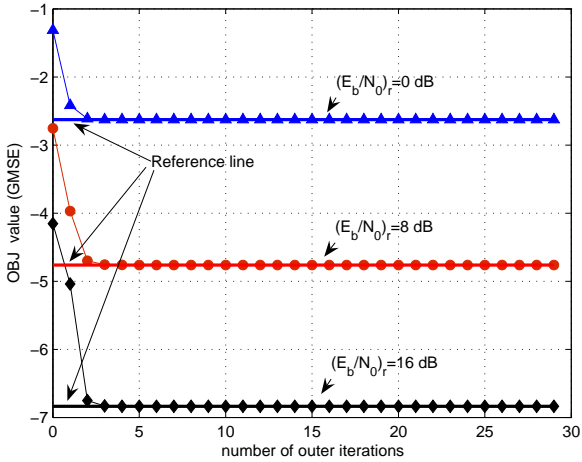}
    \caption{Objective function value (GMSE criterion) versus number of outer iterations.}
{\label {}}
\end{figure}

In Fig.~4, we show the values of the objective function, ${\rm OBJ}$, for FD-DFE, cf.~(\ref{dfe88}), for different values of $N_{fb}$. Note that ${\rm OBJ}$ for $N_{fb}=0$ serves as the upper bound, ${\rm OBJ}_{ub}$, for
the general objective function. From the figure, we observe that the upper bound is very close to the objective function for all considered values of $N_{fb}$, especially for medium to high SNR and for the
$\{2,3,3,3\}$ system. Therefore, ${\rm OBJ}_{ub}$ constitutes a good approximation for the objective function for the FD-DFE receiver.

\begin{figure}[tp]
    \centering
    \includegraphics[width=3.5in]{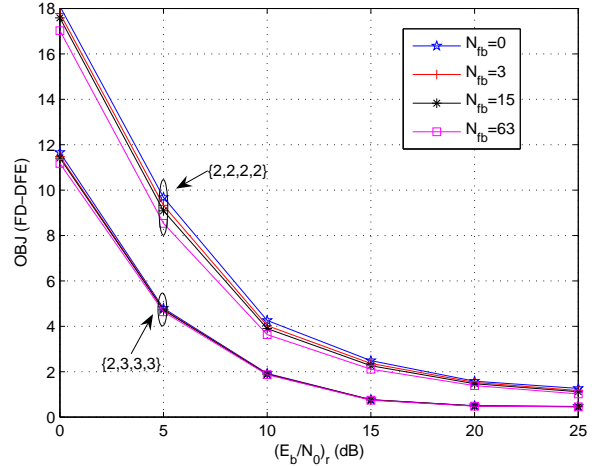}
    \caption{Objective function value of the FD-DFE receiver for different values of $N_{fb}$.}
{\label {}}
\end{figure}

\subsection{Comparison of SC-FDE and OFDM for JSR Precoding}
In Fig.~5, we show the BER of uncoded
quaternary phase-shift keying (QPSK) as
a function of $(E_b/N_0)_r$ for the
proposed FD-LE based MIMO relay system
for the three considered precoding
matrix optimization criteria. For FD-DFE, only the
GMSE criterion is considered as for FD-DFE all three
criteria are equivalent. For
comparison, the performance of a
MIMO-OFDM relay system optimized under
the same criteria is also included
\cite{Rong}. The figure shows that for
the $\{2,2,2,2\}$ system, the proposed
MIMO relay system with an FD-LE
receiver outperforms the corresponding
OFDM-based system by a large margin
since, in contrast to uncoded OFDM,
FD-LE is able to exploit the frequency
diversity offered by the channel. In
addition, for both FD-LE and OFDM, the
system employing the maxMSE criterion
offers the best performance since the
worst-case MSE is minimized. For
FD-DFE, the performance improvement
compared to FD-LE and OFDM is
remarkable and a much higher diversity
gain is observed.  On the other hand,
for the $\{2,3,3,3\}$ system, we
observe that the performance gaps
between FD-DFE, FD-LE, and OFDM become
smaller. Surprisingly, using the maxMSE
criterion, the optimized OFDM and FD-LE
systems achieve a performance very
close to that of FD-DFE. This is due to
fact that the additional antennas offer
additional spatial diversity which
helps OFDM and FD-LE to effectively
avoid the deep spectrum nulls that
otherwise negatively affect their
performance in frequency-selective
fading.

\begin{figure}[tp]
    \centering
    \includegraphics[width=3.5in]{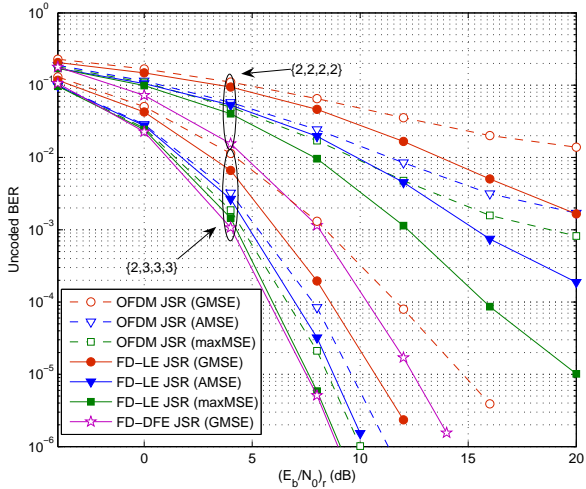}
    \caption{Uncoded BER of $\{2,2,2,2\}$ and $\{2,3,3,3\}$ MIMO relay systems for JSR precoding
using different optimization criteria.}
{\label {}}
\end{figure}

In Fig.~6, we investigate the capacity
of the OFDM and SC-FDE systems with
different optimization criteria.
As expected, the systems
optimized under the GMSE criterion have
the best performance since minimizing
the GMSE is equivalent to maximizing
the capacity. In general, the capacity
achieved by the considered MIMO-OFDM
relay systems is higher than that of
the corresponding FD-LE relay systems,
except for the case when both systems
are optimized based on the maxMSE
criterion. Indeed, the OFDM system
optimized under the maxMSE criterion
suffers from the worst capacity
performance among all the considered
schemes since the available power is
mainly used to improve the MSE of the
subcarriers with bad channel conditions
instead of taking advantage of the
subcarriers with good channel
conditions. In addition, Fig.~6 shows
that for FD-LE, the AMSE and maxMSE
criteria lead to exactly the same
capacity, which implies that the
unitary rotation of the source
precoding matrix does not influence the
capacity of the system. Furthermore, the
capacity achieved with FD-DFE is larger than
that achieved with any of the FD-LE schemes
and very close to that of OFDM. This is due
to the lower stream MSEs of FD-DFE compared
to FD-LE, which translates into larger stream SINRs
and larger system capacity. For the
$\{2,3,3,3\}$ system, we observe that
FD-LE and OFDM achieve almost the same
performance for the AMSE and GMSE
criteria, implying that with more
source/relay/destination antennas,
FD-LE will approach the achievable
capacity of the OFDM system.
\begin{figure}[tp]
    \centering
    \includegraphics[width=3.5in]{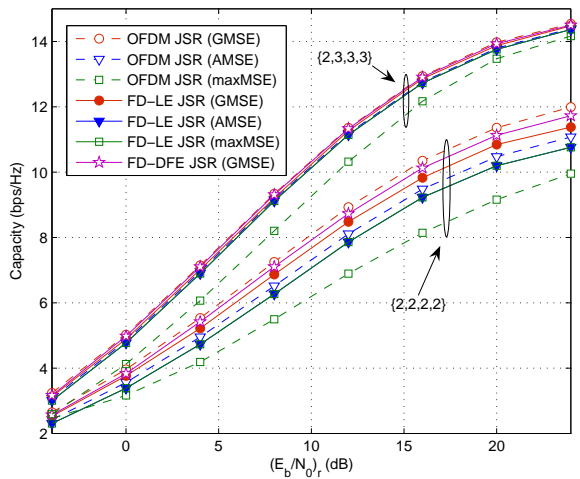}
    \caption{Capacity of $\{2,2,2,2\}$ and $\{2,3,3,3\}$ MIMO relay systems with JSR relay precoding for different optimization criteria.}
{\label {}}
\end{figure}

\subsection{Performance of Suboptimal Power Allocation Schemes}
In Figs.~7 and 8, we plot the uncoded
and coded BERs for the suboptimal power
allocation schemes discussed in Section
IV-D for a \{2,2,2,2\} system,
respectively. For the coded case, the
standard rate-1/2 convolution code with
generator matrix $(133,171)_{oct}$ is
adopted. The OFDM and FD-LE systems are
both optimized under the maxMSE
criterion. The FD-DFE system is
optimized under the GMSE criterion
since for FD-DFE all three considered
criteria are equivalent to the GMSE
criterion.
 From Fig.~7 we observe that for uncoded transmission,
 the FD-LE system outperforms the OFDM system if both employ the same precoding technique.
Fig.~7 also shows that for FD-LE and OFDM, EPA-S and ROP suffer from a considerable
 performance degradation compared to JSR, while for FD-DFE,
 the performance loss is relatively small for UPS and
 almost negligible for
EPA-S. By coding across subcarriers, OFDM-based relay
 systems can also exploit the frequency diversity of the
  channel and significantly improve their BER performance,
   cf.~Fig.~8. Nevertheless, the coded FD-LE system still
outperforms the OFDM system if the same
precoding technique
 is assumed in both cases. Also, Fig.~8 reveals that channel
 coding significanlty reduces the performance loss caused by
 suboptimal precoding techniques for
both OFDM and FD-LE.

\begin{figure}[tp]
    \centering
    \includegraphics[width=3.5in]{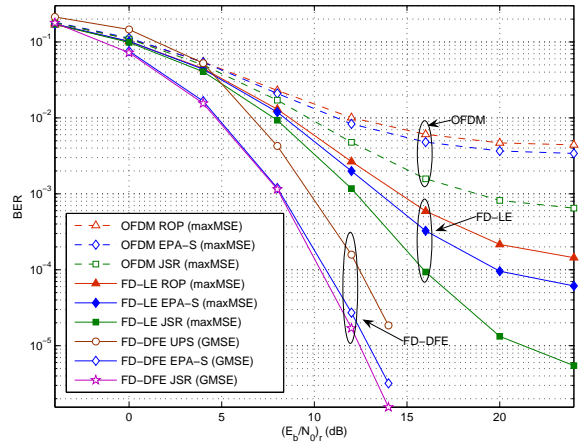}
    \caption{Uncoded BER of a $\{2,2,2,2\}$ MIMO relay system with JSR and suboptimal precoding schemes.} {\label {}}
\end{figure}

\begin{figure}[htp]
    \centering
    \includegraphics[width=3.5in]{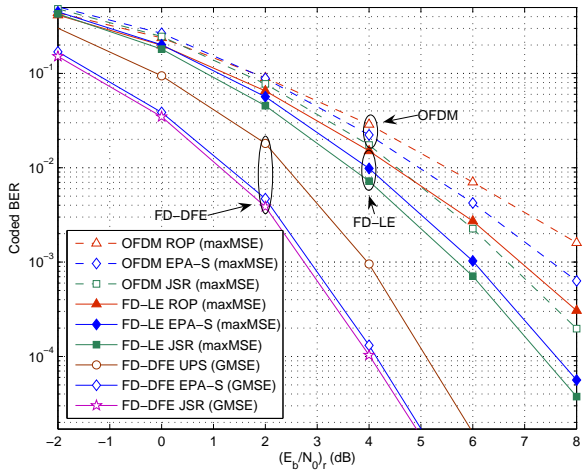}
    \caption{Coded BER of a $\{2,2,2,2\}$ MIMO relay system with JSR and suboptimal precoding schemes.} {\label {}}
\end{figure}

Since the performance of FD-DFE depends
on the the number of FBF taps, in
Fig.~9, we investigate the influence of
$N_{fb}$ on the performance of a
\{2,2,2,2\} system. The results show
that while the value of $N_{fb}$ has
limited impact on the performance of
EPA-S, it does play a critical role for
UPS. The reason is that for EPA-S, the
equivalent $S$-$R$-$D$ channel is
diagonalized into $M$ parallel
channels, thus eliminating the
inter-stream interference at the
receiver. However, for the case of UPS
, the equivalent end-to-end channel is
not fully diagonalized and the received
symbols experience inter-stream
interference. Consequently, a FBF with
sufficiently large $N_{fb}$ is required
to cancel out this interference. As can
be inferred from Fig.~9, there is a
complexity tradeoff between the
transmitter and the receiver for
FD-DFE. For EPA-S, since a small number
of FBF taps (e.g., $N_{fb}=3$) is
sufficient to achieve good performance,
the receiver complexity is similar to
that of FD-LE. However, comparatively
complex FD signal processing has to be
carried out at the transmitter. This
characteristic makes EPA-S suitable for
the downlink transmission. For the UPS
scheme, on the other hand, the transmit
processing is very simple since the
single tap precoding matrix $\qV_1$ can
be directly implemented in the TD. In
addition, the feedback overhead is low
as $\qV_1$ is identical for all
frequency tones. However, UPS requires
a longer and thus more complex FBF to
achieve a high performance. These
characteristics make UPS a very
promising scheme for uplink
transmission.
\begin{figure}[htp]
    \centering
    \includegraphics[width=3.5in]{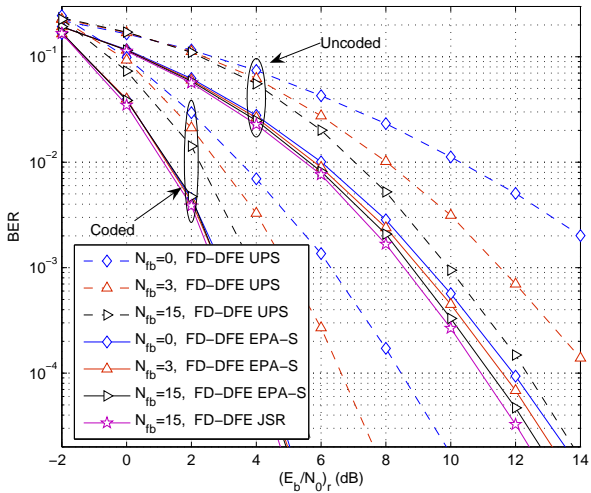}
    \caption{Uncoded and coded BER of a $\{2,2,2,2\}$ MIMO relay system with FD-DFE using JSR and suboptimal precoding schemes with different numbers of FBF taps.} {\label {}}
\end{figure}

\section{Conclusion}
In this paper, we have tackled the problem of transceiver design for MIMO relay systems
 employing SC-FDE. The optimal MMSE FD-LE and FD-DFE filters
at the destination were derived, and
the optimal structures of the source
and relay precoding matrices were
obtained in closed form for a general
family of objective functions. For
systems employing an FD-DFE receiver, we first showed
that the considered objective functions are all equivalent and
derived an upper bound on the original objective functions, which
was shown to be equal to the GMSE objective function for the
FD-LE receiver. The remaining power
allocation problem was solved globally
by using a high SNR approximation of
the objective function and efficient
convex optimization methods. Our
results show that the proposed SC-FDE
relaying schemes outperform the
corresponding OFDM schemes in terms of
both coded and uncoded BER. In
addition, the performance gap between
SC-FDE and OFDM relay systems decreases
when the number of
source/relay/destination antennas is
larger than the number of data streams.
Furthermore, we have shown that the
proposed suboptimal power allocation
schemes can reduce the system
complexity and feedback overhead at the
expense of a moderate performance
degradation, especially in case of
coded transmission, making them
promising candidates for practical
relay systems.

\section*{Appendix}

We first provide some relevant
definitions and lemmas that will be
used in the proof.

\textit{Definition 1 [15,1.A.1]:} Given
two $N \times 1$ real vectors $\qx,\qy
\in \mathbb{R}^{N}$. Let
$x_{[1]},\cdots,x_{[N]}$ and
$y_{[1]},\cdots,y_{[N]}$  denote the
components of $\qx$ and $\qy$ in
decreasing order. Then, $\qx$ is
majorized by $\qy$, or $\qx \prec \qy$,
if $\sum_{i=1}^{k}x_{[i]}\leq
\sum_{i=1}^{k}y_{[i]}$ for $k<N$ and
$\sum_{i=1}^{N}x_{[i]}=
\sum_{i=1}^{N}y_{[i]}$. Vector $\qx$ is
weakly majorized by $\qy$, or $\qx
\prec_w \qy$,  if
$\sum_{i=1}^{k}x_{[i]}\leq
\sum_{i=1}^{k}y_{[i]}, \forall k$.

\textit{Definition 2 [15,3.A.1]:} A
real function $f$ is Schur-convex if
for $\qx \prec \qy$, we have
$f(\qx)\leq f(\qy)$. Similarly, $f$ is
Schur-concave if for $\qx \prec \qy$,
we have $f(\qx)\geq f(\qy)$.

\textit{Lemma 1 [15,9.B.1]:} For a
Hermitian matrix $\qA$ with
$\mbox{diag}[\qA]$ and $\lambda(\qA)$
denoting vectors containing the main
diagonal elements and the eigenvalues
of $\qA$ arranged in decreasing order,
respectively, we have
$\mbox{diag}[\qA]\prec\lambda(\qA)$.

\textit{Lemma 2 [15,9.H.2]:} For $M$
matrices $\qA_{i}\in \mathbb{C}^{N
\times N},i=1,\cdots,M$, let
$\qB=\qA_1\qA_2\cdots\qA_M$. Then,
$\mathbf\sigma(\qB)\prec_w
\mathbf\sigma(\qA_1)\odot\mathbf\sigma(\qA_2)\odot\cdots
\mathbf\sigma(\qA_M)$, where
$\mathbf\sigma(\qX)$ denotes the vector
containing the singular values of
matrix $\qX$ arranged in decreasing
order and $\odot$ denotes the
element-wise product of two vectors.

\textit{Lemma 3 [15,3.A.8]:} A real
function $f$ satisfies $\qx \prec_w \qy
\Rightarrow f(\qx)\leq f(\qy)$ if and
only if $f$ is Schur-convex and
increasing.

\textit{Lemma 4 [15,9.H.1]:} For two
Hermitian positive semidefinite
matrices $\{\qA,\qB\} \in \mathbb{C}^{N
\times N}$ with eigenvalues
$\lambda_{\qA,i},\lambda_{\qB,i}$,
arranged in the same order, we have
${\rm tr}(\qA\qB)\geq
\sum_{i=1}^{N}\lambda_{\qA,i}\lambda_{\qB,N-i+1}
$.

\textit{Lemma 5 [15, p.7]:} For a
vector $\qx \in \mathbb{C}^{N \times
1}$, we have
$\sum_{i=1}^{N}(x_i/N)\mathbf 1\prec
\qx$, where $\mathbf 1$ is the all-ones
vector.

\textit{Lemma 6:} For Hermitian
matrices $\qA_k$ and $\hat \qA_k$,
$k=1,2,\cdots, N_c$ , if
$\mbox{diag}[\qA_k]\prec
\mbox{diag}[\hat\qA_k]$, we have
$\mbox{diag}[\sum_{k=1}^{N_c}\qA_k]\prec\mbox{diag}[\sum_{k=1}^{N_c}\hat\qA_k]$.

\textit{Proof:} According to Definition
1, if $\mbox{diag}[\qA_k]\prec
\mbox{diag}[\hat\qA_k]$, we have
$\sum_{i=1}^{j}a_{[ki]}\leq
\sum_{i=1}^{j}\hat a_{[ki]}$ for $j<N$
and $\sum_{i=1}^{N}a_{[ki]}=
\sum_{i=1}^{N}\hat a_{[ki]}$ for $j=N$,
where $a_{[ki]}$ and $\hat a_{[ki]}$
are the $i$th diagonal entries of
$\qA_k$ and $\hat \qA_k$, respectively.
Since $a_{[ki]}$ and $\hat a_{[ki]},
\forall k,i$, are non-negative, we have
$\sum_{k=1}^{N_c}\sum_{i=1}^{j}a_{[ki]}\leq
\sum_{k=1}^{N_c}\sum_{i=1}^{j}\hat
a_{[ki]}$ for $j<N$ and
$\sum_{k=1}^{N_c}\sum_{i=1}^{N}a_{[ki]}=
\sum_{k=1}^{N_c}\sum_{i=1}^{N}\hat
a_{[ki]}$ for $j=N$. Therefore, we have
$\sum_{k=1}^{N_c}\mbox{diag}[\qA_k]$$\prec
\sum_{k=1}^{N_c}\mbox{diag}[\hat\qA_k]$,
which is equivalent to
$\mbox{diag}[\sum_{k=1}^{N_c}\qA_k]\prec\mbox{diag}[\sum_{k=1}^{N_c}\hat\qA_k]$.

\textit{Lemma 7 [15, 9.B.2]:} For a
diagonal matrix $\qD \in \mathbb{C}^{M
\times M}$, there is a unitary matrix
$\qU$ such that $\qA=\qU^\dag \qD\qU$
has identical diagonal entries equal to
${\rm tr}(\qD)/M$.

We now set out to prove the optimal
structure of the source and relay
precoding matrices when
$f(\mbox{diag}[\hat\qE])$ is a
Schur-concave increasing function
w.r.t. $\mbox{diag}[\hat\qE]$. Let us
begin with the core term in the
expression for $\hat \qE$ in
(\ref{Psi}), which is given by
\begin{align} \mathbf \Psi_k^{-1}=&
 \qI_M-\nnb \\
 &\underbrace{\sigma_s^2\qQ_k^\dag(\sigma_s^2\qQ_k\qQ_k^\dag+\sigma_v^2\qH_k\qA_k\qA_k^\dag\qH_k^\dag+
 \sigma_{u}^2\qI_{N_d})^{-1}\qQ_k}_{\mathbf \Upsilon_k},
 \end{align}
 where we employed the  matrix inversion lemma. The CV matrix can now be expressed as
 \bee
\hat\qE=\frac{\sigma_s^2}{N_c}\sum_{k=1}^{N_c}\mathbf
\Psi_k^{-1}=\frac{\sigma_s^2}{N_c}(N_c\qI_M-\sum_{k=1}^{N_c}\mathbf
\Upsilon_k).
    \eee
Since
$f(\mbox{diag}[\sigma_s^2\qI_M-\frac{\sigma_s^2}{N_c}\sum_{k=1}^{N_c}\mathbf
\Upsilon_k])$ is a Schur-concave
decreasing function
w.r.t.~$\mbox{diag}[\sum_{k=1}^{N_c}\mathbf
\Upsilon_k]$,
$-f(\mbox{diag}[\sigma_s^2\qI_M-\frac{\sigma_s^2}{N_c}\sum_{k=1}^{N_c}\mathbf
\Upsilon_k])$ is a Schur-convex
increasing function
w.r.t.~$\mbox{diag}[\sum_{k=1}^{N_c}\mathbf
\Upsilon_k]$. $\mathbf \Upsilon_k $ can
be further expressed as \bee
\label{64}\mathbf \Upsilon_k
&=&\qX_k^\dag\qY_k^\dag(\qJ_k
\qJ_k^{\dag}+\sigma_u^2\qI_{N_d})^{-1}\qY_k
\qX_k, \eee where we have used the
definitions \bee
\qX_k&=&\sigma_s\qG_k\qP_k,\quad
\qY_k=\qH_k\qA_k, \nnb \\
\qJ_k&=&\qY_k(\qX_k\qX_k^{\dag}+\sigma_v^2\qI_{N_r})^{1/2}.
\eee By rewriting $\qY_k$ in terms of
$\qX_k$ and $\qJ_k$, we obtain
$\qY_k=\qJ_k(\qX_k\qX_k^{\dag}+\sigma_v^2\qI_{N_r})^{-1/2}$.
Plugging this result into (\ref{64})
leads to
 \begin{align}\label{eq66}
\mathbf \Upsilon_k
=&\qX_k(\qX_k\qX_k^{\dag}+\sigma_v^2\qI_{N_r})^{-1/2}\qJ_k^\dag
(\qJ_k
\qJ_k^{\dag}+\sigma_u^2\qI_{N_d})^{-1}
\nnb
\\ &\qJ_k(\qX_k\qX_k^{\dag}+\sigma_v^2\qI_{N_r})^{-1/2}\qX_k.
\end{align} Now, we define the following SVDs
\begin{align} &\qX_k=\qU_X^{(k)}\mathbf
\Lambda_X^{(k)} \qV_X^{(k)},\quad
     \qJ_k=\qU_J^{(k)}\mathbf \Lambda_J^{(k)} \qV_J^{(k)},\nnb \\
     &(\qX_k\qX_k^{\dag}+\sigma_v^2\qI_{N_r})^{-1/2}= \qU_X^{(k)}(\mathbf \Lambda_X^{(k)2}
     +\sigma_v^2\qI_{M})^{-1/2}
     \mathbf \Omega^{(k)},
     \label{pwconst}
\end{align} where $\qU_X^{(k)}\in
\mathbb{C}^{N_r\times M}$,
$\qU_J^{(k)}\in \mathbb{C}^{N_d\times
M}$, $\{\qV_J^{(k)},\mathbf
\Omega^{(k)}\}\in \mathbb{C}^{M\times
N_r}$, $\{\mathbf \Lambda_X^{(k)},
\qV_X^{(k)} , \mathbf \Lambda_J^{(k)}
\} \in \mathbb{C}^{M\times M}$.The
diagonal entries of $\mathbf
\Lambda_X^{(k)}$ and $ \mathbf
\Lambda_J^{(k)}$ are both sorted in
increasing order. We can use
(\ref{pwconst}) to rewrite (\ref{eq66})
as \begin{align} \mathbf \Upsilon_k =&
\qV_X^{(k)\dag}\mathbf \Lambda_X^{(k)}
{\qQ_{2}^{(k)\dag}}(\mathbf
\Lambda_X^{(k)2}+\sigma_v^2\qI_M)^{-1/2}{{\qQ_{1}^{(k)\dag}}}
 \nnb  \\
&\times(\qI_M+\sigma_u^2\mathbf\Lambda_J^{(k)-2})^{-1}
{{\qQ_{1}^{(k)}}} (\mathbf
\Lambda_X^{(k)2}+\sigma_v^2\qI_M)^{-1/2}\nnb
\\ &\times{{\qQ_{2}^{(k)}}} \mathbf
\Lambda_X^{(k)} \qV_X^{(k)},\label{Vx}
\end{align} where $\qQ_1^{(k)}=\mathbf
\qV_{J}^{(k)}\qU_X^{(k)}$ and
$\qQ_2^{(k)}=\mathbf
\Omega^{(k)}\qU_X^{(k)}$. By applying
Lemmas 1 and 2, we obtain \begin{align}
 \mbox{diag}[\mathbf \Upsilon_k]&\prec
\mathbf\lambda(\mathbf \Upsilon_k)\nnb
\\ & \prec_{w}
\mbox{diag}[\underbrace{(\qI_M+\sigma_v^2\mathbf
\Lambda_X^{(k)-2})^{-1}}_{\qD_{1}^{(k)}}
\underbrace{(\qI_M+\sigma_u^2\mathbf
\Lambda_J^{(k)-2})^{-1}}_{\qD_{2}^{(k)}}].
\end{align} Therefore, $\mbox{diag}[\mathbf
\Upsilon_k]$ is majorized by
$\mbox{diag}[\qD_{1}^{(k)}\qD_{2}^{(k)}]$
when $\qV_{X}^{(k)}= \mathbf \Xi_1$,
$\qQ_2^{(k)}=\mathbf \Xi_2$,
$\qQ_1^{(k)}=\mathbf \Xi_3$, where
$\mathbf \Xi_i \in \mathbb{C}^{M\times
M}, \forall i$, are arbitrary diagonal
matrices with unit norm diagonal
elements. Without loss of generality,
we can choose $\mathbf \Xi_i=\qI_M,
\forall i$. Hence, we have \bee
\qV_{X}^{(k)}&=& \qI_M,\quad\mathbf
\Omega^{(k)}=\qU_X^{(k)\dag}, \quad
\mathbf \qV_{J}^{(k)}=\qU_X^{(k)\dag}.
\label{eq82} \eee According to Lemma 6,
 \bee \label{73}
\mbox{diag}[\sum_{k=1}^{N_c}\mathbf\Upsilon_k]\prec_{w}\mbox{diag}[\sum_{k=1}^{N_c}(\qD_{1}^{(k)}\qD_{2}^{(k)})]
\eee holds. From the fact that
$-f(\mbox{diag}[\sigma_s^2\qI_M-\frac{\sigma_s^2}{N_c}\sum_{k=1}^{N_c}\mathbf
\Upsilon_k])$ is a Schur-convex
increasing function
w.r.t.~$\mbox{diag}[\sum_{k=1}^{N_c}\mathbf
\Upsilon_k]$, (\ref{73}), and
Definition 2, we deduce that
$-f(\mbox{diag}[\sigma_s^2\qI_M-\frac{\sigma_s^2}{N_c}\sum_{k=1}^{N_c}\mathbf
\Upsilon_k])\leq
-f(\mbox{diag}[\sigma_s^2\qI_M-\frac{\sigma_s^2}{N_c}\sum_{k=1}^{N_c}\qD_{1}^{(k)}\qD_{2}^{(k)}])$,
which is equivalent to \begin{align}&
f(\mbox{diag}[\sigma_s^2\qI_M-\frac{\sigma_s^2}{N_c}\sum_{k=1}^{N_c}\mathbf
\Upsilon_k])\nnb \\ &\geq
f(\mbox{diag}[\sigma_s^2\qI_M-\frac{\sigma_s^2}{N_c}\sum_{k=1}^{N_c}\qD_{1}^{(k)}\qD_{2}^{(k)}]).
\label{74} \end{align} Therefore, with
the help of the matrices in
(\ref{eq82}), the value of the
objective function can be reduced to
$f(\mbox{diag}[\sigma_s^2\qI_M-\frac{\sigma_s^2}{N_c}\sum_{k=1}^{N_c}\qD_{1}^{(k)}\qD_{2}^{(k)}])$.

In the following, we derive the
structure of the optimal source and
relay precoding matrices minimizing the
power consumption of the source and the
relay. For simplicity of notation, we
only consider the case when
$N_s=N_r=N_d$. The proof can be easily
extended to the case where  $N_s$,
$N_r$, and $N_d$ have different values.
From (\ref{pwconst}), we have
$\sigma_s\qU_{G}^{(k)}
\mathbf\Lambda_{G}^{(k)}\qV_{G}^{(k)\dag}\qP_k=\qU_X^{(k)}
\mathbf \Lambda_X^{(k)} \qV_X^{(k)}$,
which can be used to express the source
power consumption at frequency tone $k$
as \begin{align} &\sigma_s^2{\rm
tr}(\qP_k\qP_k^\dag)\nnb \\
&={\rm
tr}(\mathbf\Lambda_{G}^{(k)-1}\underbrace{\qU_{G}^{(k)\dag}
\qU_X^{(k)}}_{\qQ_3^{(k)\dag}} \mathbf
\Lambda_X^{(k)2}\underbrace{\qU_X^{(k)\dag}\qU_{G}^{(k)}}_{\qQ_3^{(k)}}
 \mathbf\Lambda_{G}^{(k)-1})\nnb \\
 &\geq {\rm tr}( \mathbf{\bar\Lambda}_G^{(k)-2} \mathbf
\Lambda_X^{(k)2}), \label{eq84}
\end{align} where the inequality
follows from Lemma 4, and the diagonal
matrix $\mathbf
{\bar\Lambda}_{G}^{(k)}\in
\mathbb{C}^{M \times M}$ contains the
$M$ largest singular values of
${\qG}_k$. Therefore, in order to
minimize the source power, we need to
choose $\qQ_3^{(k)}=[\qI_M\;\; \mathbf
0_{M\times(N_s-M) }]^T$, i.e.,
$\qU_X^{(k)}={\bar\qU}_{G}^{(k)}$,
where ${\bar\qU}_{G}^{(k)}$ contains
the $M$ right-most columns of
$\qU_{G}^{(k)}$. Recalling from
(\ref{eq84}) that $\qV_X^{(k)}=\qI_M$,
the source matrix can be expressed as
 \bee
\qP_k=\frac{1}{\sigma_s}\bar\qV_{G}^{(k)}\mathbf{\bar\Lambda}_{G}^{(k)-1}
\mathbf\Lambda_X^{(k)}=\bar\qV_{G}^{(k)}\mathbf
\Lambda_{P}^{(k)}, \label{PP} \eee
where $\bar\qV_{G}^{(k)}$ contains the
$M$ right-most columns of
$\qV_{G}^{(k)}$  and $\mathbf
\Lambda_{P}^{(k)}\triangleq
\frac{1}{\sigma_s}
\mathbf{\bar\Lambda}_{G}^{(k)-1}
\mathbf\Lambda_X^{(k)}$. Next, from
(\ref{pwconst}) we obtain
$\qU_{H}^{(k)}\mathbf\Lambda_{H}^{(k)}\qV_{H}^{(k)
\dag}\qA_k=\qJ_k(\qX_k\qX_k^\dag+\sigma_v^2\qI_{N_d})^{-1/2}$.
Using this result, we can express the
relay power consumption at frequency
tone $k$ as
\begin{align}
&{\rm
tr}(\qA_k(\qX_k\qX_k^\dag+\sigma_v^2\qI_{N_r})\qA_k^\dag)
\nnb \\
&={\rm
tr}(\mathbf\Lambda_{H}^{(k)-1}\qU_{H}^{(k)\dag}
\qJ_k\qJ_k^\dag
\qU_{H}^{(k)}\mathbf\Lambda_{H}^{(k)-1})\nnb \\
&={\rm
tr}(\mathbf\Lambda_{H}^{(k)-1}\underbrace{\qU_{H}^{(k)\dag}
\qU_{J}^{(k)}}_{\qQ_{4}^{(k)\dag}}\mathbf\Lambda_{J}^{(k)2}\underbrace{\qU_{J}^{(k)\dag}
\qU_{H}^{(k)}}_{\qQ_{4}^{(k)}}\mathbf\Lambda_{H}^{(k)-1})\nnb
\\ &\geq {\rm tr}(\bar{\mathbf\Lambda}_{H}^{(k)-2}\mathbf\Lambda_J^{(k)2}),
\label{inq1}
\end{align}
where the inequality follows from Lemma
4, and the diagonal matrix $\mathbf
{\bar\Lambda}_{H}^{(k)}\in
\mathbb{C}^{M \times M}$ contains the
$M$ largest singular values of
${\qH_k}$. In (\ref{inq1}), equality
holds for $\qQ_{4}^{(k)}=[\qI_M$
\;\;$\mathbf 0_{M\times(N_r-M) }]^T$,
i.e., $\qU_J^{(k)}=\bar\qU_{H}^{(k)}$,
where $\bar\qU_{H}^{(k)}$ contains the
$M$ right-most columns of
$\qU_{H}^{(k)}$. Recalling from
(\ref{eq82}) that
${\qU_X^{(k)}}^\dag=\mathbf
\Omega^{(k)}=\qV_{J}^{(k)}$, and from
(\ref{eq84}) that
$\qU_X^{(k)}=\bar\qU_{G}^{(k)}$, we
obtain for $\qA_k$ the expression
 \bee
\qA_k&=&\bar\qV_{H}^{(k)}\mathbf
{\bar\Lambda}_H^{(k)-1} \mathbf
\Lambda_J^{(k)}(\mathbf
\Lambda_X^{(k)2}+\sigma_v^2\qI_M)^{-1/2}\bar\qU_{G}^{(k)\dag}\nnb
\\ &=&\bar\qV_{H}^{(k)}\mathbf
\Lambda_{A}^{(k)}\bar\qU_{G}^{(k)\dag},
\label{AA} \eee where
$\bar\qV_{H}^{(k)}$ and
$\bar\qU_{G}^{(k)}$ contain the $M$
right-most columns of $\qV_{H}^{(k)}$
and $\qU_{G}^{(k)}$, respectively, and
$\mathbf \Lambda_{A}^{(k)}\triangleq
\mathbf {\bar\Lambda}_H^{(k)-1} \mathbf
\Lambda_J^{(k)}(\mathbf\Lambda_X^{(k)2}+\sigma_v^2\qI_M)^{-1/2}$.
Hence, we have proved that the
expressions for the source and relay
precoding matrices given in (\ref{PP})
and (\ref{AA}) minimize the objective
function $f(\mbox{diag}[\hat\qE])$,
cf.~(\ref{74}), as well as the power
consumption at the source and the
relay, cf.~(\ref{eq84}) and
(\ref{inq1}).

Now, we turn our attention to the case
when $f(\mbox{diag}[\hat\qE])$ is a
Schur-convex increasing function
w.r.t.~$\mbox{diag}[\hat\qE]$.  From
Lemma 5 we know that $ \frac{1}{M}{\rm
tr}(\hat\qE)\mathbf 1 \prec
\mbox{diag}[\hat\qE]$. Combining this
fact and Definition 2, we obtain the
inequality \bee
f(\mbox{diag}[\hat\qE])\geq
f(\frac{1}{M}{\rm tr}(\hat\qE)\mathbf
1), \eee where equality holds when the
diagonal entries of $\hat\qE$ are all
equal to $\frac{1}{M}{\rm
tr}(\hat\qE)$. In the following, we
show that by applying a unitary
rotation to the source precoding
matrix, we can achieve this equality.
Using the eigenvalue decomposition
$\qQ_k^\dag\qK_k^{-1}\qQ_k=\qU_E^{(k)}\mathbf
\Lambda_E^{(k)}\qU_E^{(k)\dag}$, we can
write $ \mathbf
\Psi_k^{-1}=\qU_E^{(k)}(\qI_M+\mathbf
\Lambda_E^{(k)})^{-1}\qU_E^{(k)\dag}$.
Let us consider the feasible source
precoding matrix
$\bar\qP_k=\qP_k\qU_E^{(k)}\qV_0$,
where $\qV_0$ is a unitary matrix and
thus does not affect the power
constraints. Replacing $\qP_k$ with
$\bar\qP_k$ in $\mathbf{\Psi}_k^{-1}$,
we arrive at \bee \mathbf{\Psi}_k^{-1}=
\qV_0^\dag(\qI_M+\mathbf\Lambda_E^{(k)})^{-1}\qV_0,
\eee which allows us to express the
error CV matrix as \bee
\hat\qE=\frac{\sigma_s^2}{N_c}\mathbf
\sum_{k=1}^{N_c}\mathbf{\Psi}_k^{-1}=\frac{\sigma_s^2}{N_c}\qV_0^\dag\sum_{k=1}^{N_c}(\qI_M+\mathbf
\Lambda_E^{(k)})^{-1}\qV_0. \eee Since
$\sum_{k=1}^{N_c}(\qI_M+\mathbf
\Lambda_E^{(k)})^{-1}$ is the sum of
$N_c$ diagonal matrices, it is also a
diagonal matrix. Based on Lemma 7, we
conclude that there exists a unitary
matrix $\qV_0$ such that $\hat\qE$ has
identical diagonal elements given by
$\frac{1}{M}{\rm tr}(\hat\qE)$. Since
the objective function is an increasing
function w.r.t.~its arguments,
minimizing the original Schur-convex
objective function is now equivalent to
minimizing $\frac{1}{M}{\rm
tr}(\hat\qE)$, which is a Schur-concave
function. Therefore, the optimal
structures of $\qP_k$ and $\qA_k$ are
given by (\ref{PP}) and (\ref{AA}),
respectively. Furthermore, as the
resulting $\qU_E^{(k)}$ can be shown to
be an identity matrix, the source
precoding matrix for Schur-convex
functions
 is given by $\bar\qP_k=\qP_k\qV_0=\bar\qV_{G}^{(k)}\mathbf \Lambda_{P}^{(k)}\qV_0$.

{}

\end{document}